\documentclass[sigconf]{aamas}  

\AtBeginDocument{%
  \providecommand\BibTeX{{%
    \normalfont B\kern-0.5em{\scshape i\kern-0.25em b}\kern-0.8em\TeX}}}
    
\usepackage{booktabs}    
\usepackage{flushend} 

\setcopyright{ifaamas}  
\copyrightyear{2020} 
\acmYear{2020} 
\acmDOI{} 
\acmPrice{} 
\acmISBN{} 
\acmConference[AAMAS'20]{Proc.\@ of the 19th International Conference on Autonomous Agents and Multiagent Systems (AAMAS 2020)}{May 9--13, 2020}{Auckland, New Zealand}{B.~An, N.~Yorke-Smith, A.~El~Fallah~Seghrouchni, G.~Sukthankar (eds.)}  

\usepackage{geometry}
\usepackage{amsmath,amssymb,xspace}
\usepackage{etex}
\usepackage{centernot}
\usepackage{stmaryrd}
\usepackage{wrapfig}
\usepackage{amsfonts}
\usepackage{mathtools}
\usepackage{xspace}
\usepackage{color}
\definecolor{lgray}{gray}{0.9}
\usepackage{fixltx2e}
\usepackage{fancybox}
\usepackage{ifthen}
\usepackage{pgf}
\usepackage{tikz}
\usetikzlibrary{arrows,automata}
\usepackage{graphics} 
\usepackage{subfig}
\usepackage{proof}
\usepackage{hyperref}
\usepackage{soul}
\usepackage{accsupp} 
\usepackage{cancel}
\usepackage{textcomp}
\usepackage{calrsfs}
\usepackage{paralist}
\usepackage{enumitem}

\usepackage{changebar}
\usepackage{todonotes}
\usepackage[super]{nth}

\DeclareMathAlphabet{\pazocal}{OMS}{zplmf}{m}{n}
\newcommand{\mcal}[1]{\pazocal{#1}}
\newcommand{\bcal}[1]{\mathcal{#1}}

\newcommand{\rulename}[1]{\mbox{\textsc{#1}}}

\newcommand{\qt}[1]{``{#1}"}

\newcommand{\rTo}[1]{\xrightarrow{#1}}

\newcommand{\true}{{\sf true}}
\newcommand{\false}{{\sf false}}

\newlength{\arrow}
\newcounter{sqindex}

 \newcommand{\rom}[1]{ \textup{(\lowercase\expandafter{\romannumeral#1})}}
 
\usepackage[ruled]{algorithm2e} 

\SetAlFnt{\small}
\SetAlCapFnt{\small}
\SetAlCapNameFnt{\small}
\SetAlCapHSkip{0pt}
\IncMargin{-\parindent}

\newcommand \Until      {\mathbin{\mcal{U}\kern-.1em}}
\newcommand \Release     {\mathbin{\mcal{R}\kern-.1em}}
\newcommand \WaitFor    {\mathbin{\mcal{W}}}
\newcommand \Since      {\mathbin{\mcal{S}\kern-.08em}}

\newcommand \g    {\mathsf{{G}\kern.08em}}
\newcommand \f    {\mathsf{{F}\kern.08em}}
\newcommand \UntilHat   {\mathbin{\LTLhat{\mcal{U}}\kern-.1em}}

\newcommand \SinceHat   {\mathbin{\LTLhat{\mcal{S}}\kern-.08em}}

\newcommand \Impl       {\mathbin{\rightarrow}}
\newcommand \Iff        {\mathbin{\leftrightarrow}}

\renewcommand \models   {\mathrel{\raisebox{-.1em}{$\vDash$}}}

\newcommand \smodels    {\models\kern-.5em\models}
\renewcommand \phi      {\varphi}

\newcommand \ltl        {\textsc{ltl}\xspace}
\newcommand \ctl				{\textsc{ctl}\xspace}
\newcommand \ltal       {\textsc{ltol}\xspace}
\newcommand \pdl       {\textsc{pdl}\xspace}
\newcommand \pspace       {\textsc{pspace}\xspace}
\newcommand \expspace       {\textsc{expspace}\xspace}

\newcommand \nlog       {\textsc{nlog}\xspace}
\newcommand \B        	{\mcal{B}\xspace}
\newcommand \bb        	{{\Bbb B}\xspace}

\newcommand{\set}[1]{\{{#1}\}}
\newcommand{\conf}[1]{\langle{#1}\rangle}

\newcommand{\band}[3]{\bigwedge\limits_{#1}^{#2}{#3}}
\newcommand{\bor}[3]{\bigvee\limits_{#1}^{#2}{#3}}

\newcommand{\bpcup}[2]{\biguplus\limits_{#1}{#2}}
\newcommand{\func}[3]{{#1}:{#2}\rightarrow{#3}}
\newcommand{\pfunc}[3]{{#1}^{#2}_{#3}}
\newcommand{\such}{{\bf\mathsf{s.t.}}\ }
\newcommand{\toall}{\star}
\newcommand{\cstate}[2]{{s}^{#1}_{#2}}
\newcommand{\sstate}[1]{{s}_{#1}}
\newcommand{\vardom}[1]{\mathsf{Dom}({#1})}
\newcommand{\vdom}{\vardom{v}}
\newcommand{\compdom}[1]{\prod_{v\in V_{#1}}\vdom}

\newcommand{\sdom}{\prod_{i}\compdom{i}}

\newcommand{\tuple}[1]{\left(#1\right)}
\newcommand{\anglebr}[1]{\left [#1\right] }
\def\<#1>{\mathinner{\langle#1\rangle}}
\newcommand{\bnxt}[1]{\Biggl<{#1}\Biggr>}
\newcommand{\pred}{\pi}
\newcommand{\sysvar}{\mathcal{V}}

\newcommand{\chan}{ch}
\newcommand{\schan}{\rulename{ch}} 

\newcommand{\id}{i}

\newcommand{\cv}{cv}
\newcommand{\scv}{\rulename{cv}} 

\newcommand{\dat}{d}
\newcommand{\datfun}{{\bf d}}
\newcommand{\sdat}{\rulename{d}} 

\newcommand{\coma}{,\ }

\newcommand{\trans}{\mcal{T}}

\newcommand{\exis}[1]{\bullet^{\exists}{#1}}
\newcommand{\all}[1]{\bullet^{\forall}{#1}}
\newcommand{\nxt}[1]{\langle{#1}\rangle}
\newcommand{\alws}[1]{[{#1}]}

\newcommand{\tracevar}[1]{\rho_{#1}}
\newcommand{\transmain}{\delta_{\phi}}

\newcommand{\transaux}{f}
\newcommand{\size}[1]{|{#1}|}
\newcommand{\Exp}[1]{2^{#1}}
\newcommand{\dexp}[1]{2^{2^{#1}}}
\newcommand{\bigo}[1]{\bcal{O}({#1})}
\newcommand{\lang}[1]{L_{\omega}({#1})}

\newcommand{\msf}[1]{\mathsf{#1}}

\newcommand{\rcp}{{{\rulename{ReCiPe}}}\xspace}

\renewcommand{\comment}[1]{}

\newcommand{\keep}{\mbox{\sc keep}}
\newcommand{\agent}[1]{\rulename{{#1}}}
\newcommand{\lineagent}{\agent{Line}}
\newcommand{\taagent}{\agent{t1}}
\newcommand{\tbagent}{\agent{t2}}
\newcommand{\tcagent}{\agent{t3}}

\newcommand{\typecvar}{{\scriptstyle@\msf{type}}}
\newcommand{\assigncvar}{{\scriptstyle@\msf{asgn}}}
\newcommand{\readycvar}{{\scriptstyle@\msf{rdy}}}

\newcommand{\msgdvar}{\rulename{msg}}
\newcommand{\lnkdvar}{\rulename{lnk}}
\newcommand{\nodvar}{\rulename{no}}

\newcommand{\stlvar}{\msf{st}}
\newcommand{\lnklvar}{\msf{lnk}}
\newcommand{\prdlvar}{\msf{prd}}
\newcommand{\stagelvar}{\msf{stage}}
\newcommand{\stagebvar}{\msf{step}}
\newcommand{\typelvar}{\msf{ltype}}
\newcommand{\assignlvar}{\msf{lasgn}}
\newcommand{\readylvar}{\msf{lrdy}}
\newcommand{\typebvar}{\msf{btype}}
\newcommand{\assignbvar}{\msf{basgn}}
\newcommand{\readybvar}{\msf{brdy}}
\newcommand{\nolvar}{\msf{no}}
\newcommand{\val}[1]{\rulename{{\scriptsize\tt{#1}}}}
\newcommand{\taval}{\val{t1}}
\newcommand{\tbval}{\val{t2}}
\newcommand{\tcval}{\val{t3}}
\newcommand{\teamval}{\val{team}}
\newcommand{\assembleval}{\val{asmbl}}
\newcommand{\formval}{\val{form}}
\newcommand{\pendval}{\val{pnd}}
\newcommand{\enval}{\val{end}}

\newcommand{\startval}{\val{strt}}

\usepackage{fontawesome}

\begin{document}
\fancyhead{}
\title{Reconfigurable Interaction for MAS Modelling}
\titlenote{This research is funded by the ERC consolidator grant D-SynMA under 
the European Union's Horizon 2020 research and innovation programme (grant 
agreement No 772459).}

\author{Yehia Abd Alrahman}
\orcid{0000-0002-4866-6931}
\affiliation{%
    \institution{University of Gothenburg}
  \city{Gothenburg} 
  \state{Sweden} 
}
\author{Giuseppe Perelli}
\orcid{0000-0002-8687-6323}
\affiliation{%
  \institution{Sapienza University of Rome}
  \city{Rome} 
  \state{Italy} 
}

\author{Nir Piterman}
\orcid{0000-0002-8242-5357}
\affiliation{%
  \institution{University of Gothenburg \and University of Leicester}
  \city{Gothenburg} 
  \state{Sweden} 
}
\begin{abstract}
  We propose a formalism to model and reason about multi-agent systems.
        We allow agents to interact and communicate in different modes
        so that they
        can pursue joint tasks;
	agents may dynamically synchronize, exchange data, adapt their 
	behaviour, and reconfigure their communication interfaces.
	The formalism defines a local behaviour based on shared variables and a 
	global one based on message passing.
	We extend \ltl to be able to reason explicitly about the intentions
	of the different agents and their interaction protocols. 
	We also study the complexity of satisfiability and 
	model-checking of this extension.
\end{abstract}

%

\keywords{Agent Theories and Models, Logics for Agent
  Reasoning, Verification of Multi-Agent Systems} 

\maketitle


\section{Introduction}\label{sec:intro}

In recent years formal modelling of multi-agent systems (MAS) and their
analysis through model checking has received much attention~\cite{Woo02,LQR17}.
Several mathematical formalisms have been suggested to represent the
behaviours of such systems and to reason about the strategies that
agents exhibit \cite{LQR17,AHK02}. 
For instance, 
modelling languages, such as RM \cite{AH99b,GHW17} and ISPL \cite{LQR17}, are used to
enable efficient analysis by representing these systems through the usage of
BDDs.
Temporal logics have been also extended and adapted (e.g., with Knowledge support
\cite{FHMV95} and epistemic operators \cite{GochetG06}) specifically to support
multi-agent modelling \cite{GL13}.
Similarly, logics that support reasoning about the intentions and
strategic abilities of such agents have been used and extended
\cite{CHP10,MMPV14}.

These works are heavily influenced by the formalisms used for verification
 (e.g., Reactive Modules~\cite{AH99b,AG04},
concurrent game structures~\cite{AHK02}, and interpreted
systems~\cite{LQR17}).
They 
rely on shared memory to implicitly model
interactions.
It is generally agreed that explicit message passing is more
appropriate to model interactions among distributed agents because of its
scalability~\cite{cougaar,mno}.
However, the mentioned 
formalisms trade the advantages of message passing for abstraction,
and abstract message exchange by controlling the visibility of state
variables of the different agents.
Based on an early result, where a compilation from shared memory to message
passing was provided \cite{stop}, it was believed that a shared memory model 
is a higher level abstraction of distributed systems. 
However, this result holds only in specific cases and
under assumptions that practically proved to be unrealistic, see~\cite{mm}.
Furthermore, the compositionality of shared memory approaches is 
limited and the supported interaction interfaces are in general not
very flexible~\cite{BBDM19}. 
Alternatively, message passing formalisms~\cite{pi1} are very
compositional and support flexible interaction interfaces.
However, unlike shared memory formalisms, they do not accurately
support awareness capabilities, where an agent may instantaneously
inspect its local state and adapt its behaviour while interacting%

To combine the benefits of both approaches recent developments~\cite{APUS,mm} 
suggest adopting hybrids, to accurately represent actual distributed systems, e.g., ~\cite{rdma,MathewsVCOBD15}. 
We propose a hybrid formalism to model and reason about
distributed multi-agent systems.
A system is represented as a set of agents (each with local state),
executing concurrently and
only interacting by message exchange. 
Inspired by multi-robot systems, 
e.g., Kilobot~\cite{kilobot} and Swarmanoid~\cite{swarm},
agents are additionally able to sense their local states and
partially their surroundings.
Interaction is driven by message passing following 
the interleaving semantics of~\cite{pi1};
in that only one agent may send a message at a time while other agents
may react to it.
To support meaningful interaction among agents \cite{WK16},
messages are not mere synchronisations, but carry data that might
be used to influence the behaviour of receivers. 

Our message exchange is adaptable and reconfigurable.
Thus, agents determine how to communicate and with whom.
Agents interact on links that change their utility
based on the needs of interaction at a given stage.
Unlike existing message-passing mechanisms, which use static 
notions of network connectivity to establish interactions, our mechanisms 
allow agents to specify receivers 
using logical formulas.
These formulas are interpreted over the evolving local
states of the different agents and thus provide a natural way to establish 
reconfigurable interaction interfaces (for example, limited range
communication \cite{MathewsVCOBD15}, messages destined for particular agents \cite{info19}, etc.).

The advantages of our formalism are threefold. We provide more realistic
 models that are close to their distributed implementations, and how actual 
 distributed MAS are developed, e.g.,~\cite{swarmdesign}. We provide a 
 modelling convenience for high level interaction features of MAS
  (e.g., coalition formation, collaboration, self-organisation, etc),
  that otherwise have to be hard-coded tediously in existing formalisms.  
Furthermore, we decouple the individual behaviour of agents from their interaction
protocols to facilitate reasoning about either one separately.
%

In addition, we extend \ltl to characterise messages and their
targets.
This way we allow
reasoning about the intentions of agents in communication.
Our logic can refer directly to the interaction protocols.
Thus the interpretation of a formula 
incorporates information
about the causes of assignments to variables and the flow of the interaction
protocol.
We also study the complexity of satisfiability and
Model-checking for our logic.


\section{An informal overview}\label{sec:overview}
We use a collaborative-robot scenario to informally illustrate the
 distinctive features of our formalism and we later formalise it in Section~\ref{sec:exp}.
The scenario is based on Reconfigurable Manufacturing Systems (RMS)~\cite{rms1,rms2},
where assembly product lines coordinate autonomously with different
types of robots to produce products.

In our formalism, each agent has a local state
consisting of a set of variables whose values may change due to
either contextual conditions or side-effects of
interaction. The external behaviour of an agent is only represented by the messages it
exposes to other agents while the local one is represented by changes to its
state variables. These variables are initialised by initial conditions
and updated by send- and receive- transition relations.
In our example, a product-line agent initiates different production
procedures based on the assignment to its product variable $\qt{\msf{prd}}$,
which is set by the operator, while it controls the progress of its status variable
$\qt{\msf{st}}$ based on interactions with other robots.
Furthermore, a product-line agent is characterised:
(1) externally only by the recruitment and assembly messages it  
sends to other robots and
(2) internally by a sequence of assignments to its local variables.

Before we explain the send- and receive- transition relations and show
the dynamic reconfiguration of communication interfaces we need to
introduce a few additional features. 
We assume that there is an agreed set of \emph{channels/links} $\schan$ that includes a unique
broadcast channel $\star$.
Broadcasts have non-blocking send and blocking receive while
multicasts have blocking send and receive.
In a broadcast, receivers (if exist) may anonymously receive a
message when they are interested in its values and when they satisfy
the send guard.
Otherwise, the agent does not participate in the interaction either
because they cannot (do not satisfy the guard) or because they are not
interested (make an idle transition).
In multicast, all agents connected to the multicast channel must
participate to enable the interaction.
For instance, recruitment messages are broadcast because a line agent
assumes that there exist enough robots to join the team while assembly messages
are multicast because they require that the whole connected team is ready to
assemble the product.

Agents dynamically decide (based on local state) whether they can
use/connect-to multicast channels while the broadcast channel is always available. 
Thus,
initially, agents may not be connected to any channel, except for the
broadcast one $\msf{\star}$.
These channels may be learned using broadcast messages
and thus a structured communication interface can be built at run-time, starting
from a (possibly) flat one. 

Agents use messages to send selected data and specify
how and to whom.
Namely, the values in a message specify what is exposed to the others;
the channel specifies how to coordinate with others; and a send
guard specifies the target.
Accordingly, each message carries an assignment to a set of agreed
\emph{data variables} $\sdat$, i.e., the exposed data;  a
channel $\msf{\chan}$; and a send guard $\pfunc{g}{s}{}$.
In order to write meaningful send guards, we assume a set of
\emph{common variables} $\scv$ that each agent stores locally and
assigns its own information (e.g., the type of agent, its location,
its readiness, etc.). Send guards are expressed in terms of conditions
on these variables and are evaluated per agent based on their assigned
local values. 
Send guards are parametric to the local state of the sender
and specify what assignments to common variables a potential receiver
must have.
For example, an agent may send a dedicated link name to a selected set
of agents by assigning a data variable in the communicated message and
this way a coalition can be built incrementally at run-time.
In our RMS, the send guard of the recruitment message specifies
the types of the targeted robots while the data values expose the
number of required robots per type and a dedicated multicast link to
be used to coordinate the production. 

Targeted agents may use incoming messages to update their
states, reconfigure their interfaces, and/or adapt their behaviour.
In order to do so, however, agents are equipped with receive guards $\pfunc{g}{r}{}$; 
that might be parametrised to local variables and channels, and thus
dynamically determine if an agent is connected to a given channel. 
The interaction among different agents is then derived based on send- and receive-
transition relations.
These relations are used to decide when to send/receive a message and what are the
side-effects of interaction. 
Technically, every agent has a send and a receive transition relation.
Both relations are parameterised by the state variables of the agent,
the data variables transmitted on the message, and by the channel name.
A sent message is interpreted as a joint transition between the send
transition relation of the sender and the receive transition relations of all
the receivers.
For instance, a robot's receive guard specifies that other than
the broadcast link it is also connected to a multicast link that matches the current value of
its local variable $\qt{\msf{lnk}}$. The robot then uses its receive transition relation
to react to a recruitment message, for instance, by assigning to its $\qt{\msf{lnk}}$ the
link's data value from the message.

Furthermore, in order to send a message the following has to happen.
The send transition relation of the sender must hold on: a given state of the sender,
a channel name, and an assignment to data variables.
If the message is broadcast,
all agents whose assignments to common variables satisfy the
send guard jointly receive the message, the others discard it.
If the message is multicast, all connected agents must satisfy the
send guard to enable the transmission (as otherwise they block the
message).
In both cases, 
sender and receivers execute their send- and receive-transition
relations jointly.
The local side-effect of the message takes into account the origin
local state, the channel, and the data.
In our example, a (broadcast) recruitment 
message is received by all robots that are not assigned to other teams (assigned ones  
discard it) and as a side effect they connect to a multicast
channel that is specified in the message.
A (multicast) assembly message can only be sent when the whole 
recruited team is ready to receive (otherwise the message is blocked)
and as a side effect the team proceeds to the next production stage. 

Clearly, the dynamicity of our formalism stems from the fact that we base
interactions directly over the evolving states of the different agents 
rather than over static notions of network connectivity as of existing approaches. 

\section{Transition systems}\label{sec:bck}



A \emph{Doubly-Labeled Transition System} (DLTS) is
$\mathcal{T}=\langle \Sigma,\Upsilon,S, \allowbreak S_0,R,L\rangle$,
where $\Sigma$ is a \emph{state alphabet}, $\Upsilon$ is a
\emph{transition alphabet}, $S$ is a set of states, $S_0\subseteq S$
is a set of initial states, $R\subseteq S\times \Upsilon \times S$ is
a transition relation, and $L:S\rightarrow \Sigma$ is a labelling
function.

A \emph{path} of a transition system $\mathcal{T}$ is a maximal sequence of states and transition labels $\sigma=s_0,a_0,s_1,a_1,\ldots$ such that $s_0\in S_0$ and 
for every $i\geq 0$ we have $(s_i,a_i,s_{i+1})\in R$. 
We assume that for every state $s\in S$ there are $a\in \Upsilon$ and
$s'\in S$ such that $(s,a,s')\in R$. 
Thus, a sequence $\sigma$ is maximal if it is infinite.
If $|\Upsilon|=1$ then $\mathcal{T}$ is a \emph{state-labeled system} and if $|\Sigma|=1$ then $\mathcal{T}$ is a \emph{transition-labeled system}. 

We introduce \emph{Discrete Systems} (DS) that represent state-labeled systems symbolically.
A DS is $\mathcal{D} = \langle \mathcal{V}, \theta, \rho \rangle$, where the components of
$\mathcal{D}$ are as follows:
\begin{itemize}[label={$\bullet$}, topsep=0pt, itemsep=0pt, leftmargin=10pt]
\item
  $\mathcal{V} = \set{v_1,...,v_n}$: A finite set of typed variables.  
  Variables range over discrete domains, e.g., Boolean or Integer.
A \emph{state} $s$ is an interpretation of  $\mathcal{V}$,
  i.e., if $D_v$ is the domain of $v$, then $s$ is in
  $\prod_{v_i\in \mathcal{V}} D_{v_i}$.

\item
  $\theta$ : The \emph{initial condition}. This is an
  assertion over $\mathcal{V}$ characterising all the initial states of the
  DS.  
  A state is called \emph{initial} if it satisfies $\theta$.
\item
  $\rho$ : A \emph{transition relation}. 
  This is an assertion $\rho(\mathcal{V}\cup\mathcal{V}')$, where $\mathcal{V}'$ is a primed
  copy of variables in $\mathcal{V}$. The transition relation $\rho$
  relates a state $s\in\Sigma$ to its \emph{$\mathcal{D}$-successors}
  $s'\in\Sigma$, i.e., 
  $(s,s')\models \rho$, where $s$ is an interpretation to
  variables in $\mathcal{V}$ and $s'$ is for variables in $\mathcal{V}'$.
\end{itemize}

The DS ${\mathcal{D}}$ gives rise to a state-labeled transition system $\mathcal{T}_{\mathcal{D}}=\langle\Sigma,\{1\},T,T_0,R\rangle$, where $\Sigma$ and $T$ are the
set of states of $\mathcal{T}_\mathcal{D}$, $T_0$ is the set of
initial states, and $R$ is the set of triplets $(s,1,s')$ such that
$(s,s')\models \rho$. Clearly, the paths of ${\mathcal{T}}_\mathcal{D}$ are exactly the
 paths of $\mathcal{D}$, but the size of $\mathcal{T}_\mathcal{D}$ is 
exponentially larger than the description of $\mathcal{D}$.

A common way to translate a DLTS into a DS, which we adapt
and extend below, would be to include 
additional variables that encode the transition alphabet.
Given such a set of variables $\mathcal{V}_\Upsilon$, an assertion
$\rho(\mathcal{V} \cup \mathcal{V}_\Upsilon \cup \mathcal{V}')$ characterises
the triplets $(s,\upsilon,s')$ such that $(s,\upsilon,s') \models
\rho$, where $s$ supplies the interpretation to $\mathcal{V}$, $\upsilon$
to $\mathcal{V}_\Upsilon$ and $s'$ to $\mathcal{V}'$.

\section{\rcp: Reconfigurable Communicating Programs}\label{sec:model}

We formally present the \rcp communication
formalism and its main ingredients. We start by specifying agents (or programs)
and their local behaviours and we show how to compose these local
behaviours to generate a global (or a system) one.
We assume that agents rely on a set of common variables $\scv$, a set
of data variables $\sdat$, and a set of  channels $\schan$ containing
the broadcast one $\toall$.

\begin{definition}[Agent]\label{def:comp}
An agent is $
A_{\id}=\langle V_{\id}\coma
f_{\id}\coma \pfunc{g}{s}{\id}\coma
\pfunc{g}{r}{\id},$ $\pfunc{\trans}{s}{\id}\coma \pfunc{\trans}{r}{\id},\theta_{\id}\rangle
$, where:
\begin{itemize}[label={$\bullet$}, topsep=0pt, itemsep=0pt, leftmargin=10pt]
\item $V_{\id}$ is a finite set of typed local variables, each
  ranging over a finite domain. A state $\cstate{\id}{}$ is an
  interpretation of $V_{\id}$, i.e., if $\mathsf{Dom}(v)$ is the
  domain of $v$, then $\cstate{\id}{}$ is an element in
  $\compdom{\id}$. We use $V'$ to denote the primed copy of $V$ and  
$\mathsf{Id}_{\id}$ to denote the assertion $\bigwedge_{v\in
  V_{\id}}v=v'$.

\item $\func{f_{\id}}{\scv}{V_{\id}}$ is a function, associating common variables to local variables.
We freely use the notation $f_{\id}$ for the assertion $\bigwedge_{\cv\in \scv}\cv=f_{\id}(\cv)$.

\item $\pfunc{g}{s}{\id}(V_{\id}, \schan, \sdat, \scv)$ is a send guard
  specifying a condition on receivers. That is, the predicate, obtained
  from $\pfunc{g}{s}{\id}$ after assigning $\cstate{\id}{}$, $\chan$, and $\datfun$ (an
  assignment to $\sdat$)
  , which is checked against every receiver $j$ after
  applying $f_{j}$.

\item $\pfunc{g}{r}{\id}(V_{\id}, \schan)$ is a receive guard describing
  the connection of an agent to channel $\chan$. We let
  $\pfunc{g}{r}{\id}(V_{\id}, \toall)$ $ = \true$, i.e., every agent
  is always connected to the broadcast channel. We note, however, that
  receiving a broadcast message could have no effect on an agent. 
\item $\pfunc{\trans}{s}{\id}(V_{\id}, V'_{\id}, \sdat, \schan)$ is an assertion
describing the send transition relation while $\pfunc{\trans}{r}{\id}(V_{\id}, V'_{\id}, \sdat, \schan)$ is an assertion
describing the receive transition relation.
We assume that an agent is broadcast input-enabled, i.e.,
$\forall v, \datfun\ \exists v'\ \such$ $\pfunc{\trans}{r}{\id}(v, v',
\datfun, \toall)$.
 
\item $\theta_{\id}$ is an assertion on $V_{\id}$ describing the
  initial states, i.e., a state is initial if it satisfies
  $\theta_{\id}$.  
\end{itemize}
\end{definition}

Agents exchange messages. A message is defined
by the channel it is sent on, the data it carries, the sender identity,
and the assertion describing the possible local assignments to
common variables of receivers.
Formally:

\begin{definition}[Observation]\label{def:obsrv}
An observation is a tuple
$m=\tuple{\chan,\datfun,\id,\pred}$, where $\chan$ is a channel, $\datfun$ is
an assignment to $\sdat$, $\id$ is an identity,
and $\pred$ is a predicate over \scv.
\end{definition}
In Definition~\ref{def:obsrv} we interpret $\pred$ as a set of possible assignments to common
variables $\scv$.
In practice, $\pred$ is obtained from
$\pfunc{g}{s}{\id}(\cstate{\id}{},\chan,\datfun,\scv)$ for an agent
$\id$, where $\cstate{\id}{}\in\compdom{\id}$ and $\chan$ and $\datfun$
are the channel and assignment in the observation. 
We freely use $\pi$ to denote either a predicate over $\scv$ or its
interpretation, i.e., the set of variable assignments $c$ such that $c
\models \pi$.

A set of agents agreeing on the common variables $\scv$,
data variables $\sdat$, and channels $\schan$ define a \emph{system}.
We define a DLTS capturing the
interaction and 
then give a DS-like symbolic representation of the same system.

Let $\Upsilon$ be the set of possible observations.
That is, let $\schan$ be the set of channels, $\mathcal{D}$ the
product of the domains of variables in $\sdat$, $\bcal{A}$ the set of
agent identities, and $\Pi(\scv)$ the set of predicates over $\scv$
then
$\Upsilon \subseteq \schan \times \mathcal{D}\times \bcal{A}\times
\Pi(\scv)$.
In practice, we restrict attention to predicates in $\Pi(\scv)$ that
are obtained from $\pfunc{g}{s}{\id}(V_{\id},\schan,\sdat, \scv)$ by assigning
to $\id$ and $V_{\id}$ the identity and a state of some agent.

Let $S_{\id}$=$\Pi_{v\in V_{\id}}\mathsf{Dom}(v)$ and
$S=\Pi_{\id}S_{\id}$.
Given an assignment $s\in S$ we denote by $s_{\id}$ the projection of 
$s$ on $S_{\id}$.

\begin{definition}[Transition System]\label{def:ts}
  Given a set $\{A_{\id}\}_{\id}$ of agents,
  we define a doubly-labeled transition system $\mathcal{T}=\langle
  \Sigma,\Upsilon,S,S_0,R,L\rangle$, where $\Upsilon$ and $S$ are as
  defined above, $\Sigma=S$, $S_0$ are the states that satisfy
  $\bigwedge_{\id}\theta_{\id}$, $L:S\rightarrow \Sigma$ is the
  identity function, and $R$ is as follows.

  A triplet $(s,\upsilon,s') \in R$, where
  $\upsilon=(\chan,\datfun,\id,\pred)$, if the following conditions hold:
  \begin{itemize}[label={$\bullet$}, topsep=0pt, itemsep=0pt, leftmargin=10pt]
  \item
    For the sender $\id$ we have that
    $\pred=\pfunc{g}{s}{\id}(s_{\id},\chan, \datfun)$, i.e.,
    $\pred$ is obtained from $\pfunc{g}{s}{\id}$ by assigning the
    state of $\id$, the data variables assignment $\datfun$ and the channel $\chan$, and
    $\pfunc{\trans}{s}{\id}(s_{\id},s'_{\id},\datfun,\chan)$ evaluates to
    $\true$. 
  \item
    For every other agent $\id'$ we have that either
    (a)
    $\pred(f^{-1}_{\id'}(s_{\id'}))$, 
    $\pfunc{\trans}{r}{\id'}(s_{\id'},s'_{\id'},\datfun,\chan)$,
    and $\pfunc{g}{r}{\id'}(s_{\id'},\chan)$ all evaluate to $\true$,
    (b)
    $\pfunc{g}{r}{\id'}(s_{\id'},\chan)$ evaluates to $\false$ and
    $s_{\id'}=s'_{\id'}$,
    or
    (c) $\chan=\star$, $\pred(f^{-1}_{\id'}(s_{\id'}))$ evaluates to
    $\false$ and $s_{\id'}=s'_{\id'}$.
    By $\pred(f^{-1}_{\id'}(s_{\id}))$ we denote the assignment of $v\in\scv$ by
     $s_{\id}(f_{\id'}(v))$ in $\pred$. 
  \end{itemize}
\end{definition}
An observation $(\chan,\datfun,\id,\pred)$ labels a
transition from $s$ to $s'$ if the sender $\id$ determines the
predicate (by assigning $s_{\id}$, $\datfun$, and $\chan$ in $\pfunc{g}{s}{\id}$)
and the send transition of $\id$ is satisfied by assigning $s_{\id}$,
$s'_{\id}$ and $\datfun$ to it, i.e., the sender changes the state from
$s_\id$ to $s'_{\id}$ and sets the data variables in the observation
to $\datfun$. All the other agents either (a) satisfy this
condition on receivers (when translated to their local copies of the
common variables), are connected to $\chan$ (according to
$\pfunc{g}{r}{\id'}$), and perform a valid transition when reading the
data sent in $\datfun$, (b) are not connected to $\chan$
(according to $\pfunc{g}{r}{\id'}$) and all their variables do not
change, or (c) the channel is a broadcast channel, the agent
does not satisfy the condition on receivers, and all their variables
do not change.

We now define a symbolic version of the same transition
system. To do that we have to extend the format of the
allowed transitions from assertions over an extended set of variables
to assertions that allow quantification.

\begin{definition}[Discrete System]\label{def:sys}
Given a set $\{A_{\id}\}_{\id}$ of agents, a system is defined as follows: $
S=\conf{\sysvar\coma\rho\coma\theta}
$, where $\sysvar=\bpcup{\id}{V_{\id}}$ and $\theta=\band{\id}{}{\theta_{\id}}$ and a state of the system is in $\sdom$.  
The transition relation of the system is characterised as follows:
\[\small
\begin{array}{l}
\rho:\ \exists \chan\ \exists
\sdat\ \bor{k}{}{\pfunc{\trans}{s}{k}(V_k, V'_k, \sdat, \chan)} \wedge\band{j\neq k}{}{\exists \scv. f_j\wedge }\\
\hfill\tuple{\begin{array}{l r}
\multicolumn{2}{l}{\pfunc{g}{r}{j}(V_j, \chan)\wedge \pfunc{\trans}{r}{j}(V_j, V'_j, \sdat, \chan) \wedge\ \pfunc{g}{s}{k}(V_k, \chan, \sdat, \scv)}
\\[4pt]
\vee & 
\neg\pfunc{g}{r}{j}(V_j, \chan)\wedge \mathsf{Id}_j
\\[4pt]
\vee & 
\chan=\toall\wedge \neg \pfunc{g}{s}{k}(V_k, \chan, \sdat, \scv)\wedge \mathsf{Id}_j

\end{array}}
\end{array}
\]
\end{definition}

The transition relation $\rho$ relates a system state
$\sstate{}$ to its successors $s'$ given an
observation $m=\tuple{\chan,\datfun,k,\pred}$. 
Namely, there exists an agent $k$ that sends a message with data $\datfun$
(an assignment to $\sdat$) with assertion $\pred$ (an assignment to
$\pfunc{g}{s}{k}$) on channel $\chan$ and all other agents are either (a)
connected, satisfy the send predicate, and participate in the
interaction, (b) not connected and idle, or (c) do not satisfy the
send predicate of a broadcast and idle.
That is, the agents satisfying $\pred$ (translated to their local state
by the conjunct $\exists \scv.f_j$) and connected to channel $\chan$ 
(i.e., $\pfunc{g}{r}{j}(\cstate{j}{}, \chan)$) get the
message and perform a receive transition. As a result of interaction,
the state variables of the sender and these receivers might be
updated.
The agents that are \emph{not connected} to the channel (i.e.,
$\neg\pfunc{g}{r}{j}(\cstate{j}{}, \chan)$) do not
participate in the interaction and stay still.
In case of broadcast, namely when sending
on $\toall$, agents are always connected and the set of receivers 
not satisfying $\pred$ (translated again as above) stay still.
Thus, a blocking multicast arises when a sender is blocked until all
\emph{connected} agents satisfy $\pred\wedge f_j$.
The relation ensures that, when sending on
a channel that is different from the broadcast channel $\toall$, the set
of receivers is the full set of \emph{connected} agents.
On the broadcast channel agents who do not satisfy the send
predicate do not block the sender.

The translation above to a transition system leads to a natural
definition of a trace, where the information about channels, data,
senders, and predicates is lost. We extend this definition to
include this information as follows:

\begin{definition}[System trace]\label{def:systrace}
A system trace is an infinite sequence 
$\sstate{0}\rTo{m_0}\sstate{1}$ $\rTo{m_1}\dots$ of system states and
observations such that $\forall t\geq 0$:
$m_t=\tuple{\chan_t,\datfun_t,k,\pred_t}$,
$\pred_t=\pfunc{g}{s}{k}(\cstate{k}{t},\datfun_t, \chan_t)$, and: 
\[\small
\begin{array}{l c l l}
    \tuple{\sstate{t}
  \coma\sstate{t+1}}  \models & \hspace{-1cm}
{\pfunc{\trans}{s}{k}(\cstate{k}{t}, \cstate{k}{t+1},
\datfun_t, \chan_t)}  \wedge \band{j\neq k}{}{\exists \scv. f_j\wedge }\\
& \hspace{-.5cm}
\tuple{\begin{array}{l r}
\multicolumn{2}{l}{\pfunc{g}{r}{j}(\cstate{j}{t},\chan_t)\wedge \pfunc{\trans}{r}{j}(\cstate{j}{t}, \cstate{j}{t+1}, \datfun_t, \chan_t)\wedge\pred_t}\\[4pt]
\vee &

\neg\pfunc{g}{r}{j}(\cstate{j}{t},\chan_t)\wedge \cstate{j}{t}=\cstate{j}{t+1}\\[4pt]
\vee  &

\chan_t=\toall\wedge \neg \pred_t\wedge \cstate{j}{t}=\cstate{j}{t+1}

\end{array}}
\end{array}
\]

\end{definition}

That is, we use the information in the observation to localize the
sender $k$ and to specify the channel, data values, and the
send predicate.

The following lemma relates the traces arising from Definition~\ref{def:systrace} to that of Definition~\ref{def:ts}.
\begin{lemma}{\label{lem:dlts=ds}} The traces of a system 
composed of a set of agents $\{A_{\id}\}_{\id}$ 
  are the paths of the induced DLTS.
\end{lemma}


\section{collaborative robots}\label{sec:exp}
We complete the details of the RMS example informally described in
Section~\ref{sec:overview}.
Many aspects of the example are kept very simple on purpose due to
lack of space. 
 
The system, in our scenario, consists of an assembly product line agent 
(\val{line}) and several types of task-driven robots. We only give a
program for type-$1$ (\taval) because  type-$2$
(\tbval) and type-$3$ (\tcval) are similar.
A product line is responsible for assembling the main parts and delivering 
the final product.
Different types of robots are responsible for sub-tasks, e.g., retrieving and/or
assembling individual parts.
The product line is generic and can be used to produce different products 
and thus it has to determine the set of resources, to recruit a team of robots, 
to split tasks and to coordinate the final stage of production. 

Every agent has copies of the common variables: $\typecvar$ indicating its type
($\val{line}$, $\taval$, $\tbval$, $\tcval$),
$\assigncvar$ indicating whether a robot is assigned, and $\readycvar$
indicating what stage of production the robot is in.
The set of channels includes the broadcast channel $\toall$ and
multicast channels $\{{A},\ldots \}$.
For simplicity, we only use the multicast channel ${A}$ and fix it to the line agent.
The set of data variables is $\{\msgdvar, \nodvar, \lnkdvar\}$
indicating the type of the message, a number (of robots per type),
and a name of a channel.
When a data variable is not important for some message it is omitted
from the description. We also use the notation $\keep(v)$ to denote that a variable is not
changed by a transition.

In addition to copies of common variables (e.g., $f_l(\typecvar)$ $=\typelvar$), the  line agent has the following state variables: $\stlvar$ is a state ranging over $\{\pendval,
\startval\}$ (pending and  start), $\lnklvar$ is the link of the product line,
$\prdlvar$ is the id of the active product, and $\stagelvar$ is used
to follow stages of the production.
The initial condition ${\theta_l}$ of a line agent is:
\smallskip

$
\small
{\theta_{l}} : \stlvar=\pendval \wedge \stagelvar=\val{0}  \wedge \lnklvar=A \wedge (\prdlvar=\val{1} \vee \prdlvar=\val{2})
$\smallskip

Thus, the line agent starts with a task of assembling one of two products
and uses channel $A$.
If there are multiple lines, then each is initialised
with a different channel. 

The send guard of $\lineagent$ is of the form:
\[\small
\begin{array}{r c l}
  \pfunc{g}{s}{l} \ : \ \chan{=}\star\wedge \neg\assigncvar \wedge \hfill
 ( \prdlvar{=}\val{1}{\rightarrow}(\typecvar{=}\taval\vee \typecvar{=}\tbval))
    \wedge\\ 
    \multicolumn{3}{r}
        {(\prdlvar{=}\val{2}{\rightarrow}(\typecvar{=}\taval \vee \typecvar{=}\tcval))
        \vee
  \chan{=}\msf{lnk}\wedge \readycvar=\stagelvar}
\end{array}
\]

Namely, broadcasts are sent to robots whose $\assigncvar$ is false (i.e.,
free to join a team). If the identity of the product to be
assembled is 1, then the required agents are $\taagent$ and
$\tbagent$ and if the identity of the product is 2, then the required agents
are $\taagent$ and $\tcagent$. 
Messages on channel $A$ (the value of $\lnklvar$) are sent to connected agents
when they reach a matching stage of production, i.e., $ \readycvar=\stagelvar$.
The receive guard of $\lineagent$ is $\chan=\star$, i.e., it is only connected to channel $\star$.

The send transition relation of $\lineagent$ is of the form:
\[\small
\begin{array}{rl} 
  {\pfunc{\trans}{s}{l}}:&
  \keep(\lnklvar,\prdlvar,\typelvar,\assignlvar,\readylvar) \wedge \\
  \multicolumn{2}{r}{
    \left (
    \begin{array}{l l}
&  \stlvar=\pendval \wedge 
  \datfun(\msgdvar\mapsto \teamval;\nodvar\mapsto \val{2};\lnkdvar\mapsto\lnklvar)\\
\multicolumn{2}{r}{\wedge~ \stagelvar'=\val{1} \wedge
    \stlvar'=\startval \wedge \chan=\toall}\\[2pt]
\vee& \stlvar=\startval \wedge
  \datfun(\msgdvar\mapsto\assembleval) \wedge  \stagelvar=\val{1}\wedge\\
\multicolumn{2}{r}{\wedge~ \stlvar'=\startval \wedge \stagelvar'=\val{2}\wedge \chan=\lnklvar} \\[2pt]
\vee& \stlvar=\startval \wedge
    \datfun(\msgdvar\mapsto\assembleval) \wedge \stlvar'=\pendval \\
    \multicolumn{2}{r}{\wedge~ \stagelvar=\val{2}\wedge \stagelvar'=\val{0}\wedge \chan=\lnklvar}
    \end{array} \right ) }
\end{array}
\]
$\lineagent$ starts in the pending state (see $\theta_l$).
It broadcasts a request ($\datfun(\msgdvar\mapsto\teamval)$) for two
robots ($\datfun(\nodvar\mapsto\val{2})$) per required type asking them to join
the team on the multicast channel stored in its $\lnklvar$ variable
($\datfun(\lnkdvar\mapsto\lnklvar)$).
According to the send guard, if the identity of the product to assemble is 
1 ($\prdlvar=\val{1}$) the broadcast goes to type 1 and type 2 robots and if
the identity is 2 then it goes to type 1 and type 3 robots.
Thanks to channel mobility (i.e., $\msf{\datfun(\rulename{lnk})=lnk}$) a
team on a dedicated link can be formed incrementally at run-time.
In the start state, $\lineagent$ attempts an \rulename{assemble}
(blocking) multicast on $A$.
The multicast can be sent only when the entire team completed the work
on the production stage (when their common
variable $\readycvar$ agrees with $\stagelvar$).
One multicast increases the value of $\stagelvar$ and keeps
$\lineagent$ in the start state.
A second multicast finalises the production and $\lineagent$
becomes free again.

We set 
$\small
  \pfunc{\trans}{r}{l}\hspace{-1mm}:\keep(\msf{all})$
  as $\lineagent$'s recieve transition relation.
That is, $\lineagent$ is not influenced by incoming messages.

We now specify the behaviour of $\msf{\rulename{t1}}$-robots and show how an autonomous and 
incremental one-by-one team formation is done anonymously at run-time.
In addition to copies of common variables a $\msf{\rulename{t1}}$-robot has the following variables:
 $\stlvar$ ranges over $\set{\msf{\pendval, \startval, \val{end}}}$,
$\stagebvar$ is used to control the progress of individual behaviour,
$\msf{no}$ (resp. $\msf{lnk}$) is a placeholder to a number (resp. link) 
learned at run-time, and $f_b$ relabels common variables as follows: $f_b(\typecvar)=\typebvar$, $f_b(\assigncvar)=\assignbvar$ and $f_b(\readycvar)=\readybvar$.

Initially, a robot is available for recruitment:

$\small
\begin{array}{rcl}
{\theta_{b}} : \msf{(st=\pendval)\wedge(btype=\taval)\wedge \neg basgn\wedge(lnk=\bot)\wedge}\\
(\stagebvar=\readybvar=\nolvar=\val{0})
\end{array}
$

The send guard specifies that a robot only broadcasts to unassigned robots of the
same type, namely:

$
\small
\begin{array}{rcl}
\pfunc{g}{s}{b} : (\chan=\star)\wedge
  (\typecvar=\typebvar)\wedge \neg \assigncvar.
\end{array}
$

The receive guard specifies that a $\msf{\rulename{t1}}$-robot is connected  either to
a broadcast $\star$ or to a channel matching the value of its link variable:
$
\small
\begin{array}{rcl}
\pfunc{g}{r}{b} : \chan=\star \vee \chan=\msf{lnk}.
\end{array}
$

The send $\small\pfunc{\trans}{s}{b}$ and receive $\small\pfunc{\trans}{r}{b}$ transition relations are:
\[\small
\begin{array}{rl} 
  {\pfunc{\trans}{s}{b}}:&
  \keep(\lnklvar,\typebvar) \wedge \\
  \multicolumn{2}{r}{
    \tuple{
    \begin{array}{l l}
& 
\stlvar=\startval \wedge \datfun(\msgdvar\mapsto \formval;\lnkdvar\mapsto\lnklvar;\nodvar\mapsto\nolvar-\val{1)}\\
\multicolumn{2}{r}{\wedge~ (\nolvar\geq \val{1)}\wedge\stagebvar=\val{0}\wedge\stagebvar'=\val{1}\wedge\stlvar'=\enval
    }\\
  \multicolumn{2}{r}{\wedge~ \assignbvar'\wedge (\nolvar'=\val{0})\wedge
 \chan=\toall
  }\\[2pt]
  
\vee& \stlvar=\enval \wedge
    \stagebvar=\val{1}\wedge\stagebvar'=\val{2}\wedge\ \dots\\
&\vdots\qquad[\msf{\text{\scriptsize\rulename{individual behavior}}}]\\[2pt]
\vee& \stlvar=\enval \wedge
    \dots\wedge\stagebvar'=\val{n}\wedge\readybvar'=\val{1}

    \end{array} }}
\end{array}
\]

\[\small
\begin{array}{rl} 
  {\pfunc{\trans}{r}{b}}:&
  \keep(\typebvar) \wedge \\
  \multicolumn{2}{r}{
    \tuple{
    \begin{array}{l l}
& 
\stlvar=\pendval \wedge \datfun(\msgdvar\mapsto \teamval)\wedge\stlvar'=\startval\\
\multicolumn{2}{r}{\wedge~ \lnklvar'=\datfun(\lnkdvar)\wedge\nolvar'=\datfun(\nodvar)\wedge
 \chan=\toall
    }\\[2pt]
  
\vee& \stlvar=\stlvar'=\startval \wedge \datfun(\msgdvar\mapsto \formval)\wedge \datfun(\nodvar)>\val{0}\wedge\chan=\toall\\
\multicolumn{2}{r}{\wedge~\neg\assignbvar'\wedge\lnklvar'=\datfun(\lnkdvar)\wedge\nolvar'=\datfun(\nodvar)
 
    }\\[2pt]
\vee& \stlvar=\startval \wedge \datfun(\msgdvar\mapsto \formval;\nodvar\mapsto\val{0})\wedge
 \chan=\toall\\
\multicolumn{2}{r}{\wedge~ \neg \assignbvar' \wedge \stlvar'=\pendval\wedge\lnklvar'=\bot\wedge\nolvar'=\val{0}
    }\\[2pt]
    
    \vee& \stlvar=\enval \wedge \datfun(\msgdvar\mapsto \assembleval)\wedge
 \readybvar=\val{1}\wedge\chan=\lnklvar\\
\multicolumn{2}{r}{\wedge~ \assignbvar' \wedge \stagebvar=\val{n}\wedge\stlvar'=\enval\wedge\readybvar'=\val{2}\wedge\stagebvar'=\val{0}
    }\\[2pt]
    \vee& \stlvar=\enval \wedge \datfun(\msgdvar\mapsto \assembleval)\wedge
 \readybvar=\val{2}\wedge\chan=\lnklvar\\
\multicolumn{2}{r}{\wedge~ \stlvar'=\pendval\wedge\readybvar'=\val{0}
    \wedge \lnklvar'=\bot\wedge\neg\assignbvar'
    }
    \end{array} }}
\end{array}
\]

The team formation starts when unassigned robots are in pending states ($\pendval$) as specified by the initial condition $\theta_{b}$. From this state they may only receive a team message from a line agent (by the \nth{1} disjunct of $\pfunc{\trans}{r}{b}$). The message contains the number of required robots
$\datfun(\nodvar)$ and a team link $\datfun(\lnkdvar)$.
The robots copy these values to their local variables
(i.e., $\lnklvar'=\datfun(\lnkdvar)$ etc.) and 
move to the start state ($\stlvar'=\startval$) where they either join the team (by $\pfunc{\trans}{s}{b}$) or
step back (by $\pfunc{\trans}{r}{b}$) as follows:
\begin{itemize}[label={$\bullet$}, topsep=0pt, itemsep=0pt, leftmargin=10pt]
\item
  By the \nth{1} disjunct of $\pfunc{\trans}{s}{b}$ a robot 
  joins the team by \emph{broadcasting} a $\formval$ 
  message to $\msf{\rulename{t1}}$-robots forwarding the
  number of still required robots
  ($\datfun(\rulename{no})=(\nolvar-\val{1})$) and the team link
  ($\msf{\datfun(\rulename{lnk})=lnk}$).
  This message is sent only if $\msf{no\geq\val{1}}$, i.e, at least one 
  robot is needed.
  The robot then moves to state
  $(\enval)$ to start its mission.
\item
  By the \nth{2} disjunct of $\pfunc{\trans}{r}{b}$ a robot \emph{receives} a $\formval$ message from a 
  robot, updating the number of still required robots (i.e., if 
  ${\datfun(\nodvar)>\val{0}}$), and remains in the start state.
\item
  By the \nth{3} disjunct of $\pfunc{\trans}{r}{b}$ a robot \emph{receives} a $\formval$ message from a robot,
  informing that no more robots are needed, i.e.,
  ${\datfun(\nodvar)=\val{0}}$. The robot then moves to the pending state and disconnects from the team link, i.e.,
  $\lnklvar'=\val{$\bot$}$. Thus it may not block interaction on the
  team link. 
\end{itemize}
The
last disjuncts of $\pfunc{\trans}{s}{b}$ specify that a team robot (i.e., with
$\stagebvar=\val{1}$) starts its mission
independently until it finishes ($\stagebvar'=\val{n}\wedge\readybvar'=1$) . When all team robots finish they enable an $\assembleval$ message on $\msf{A}$, to start the next stage of production as specified by the \nth{4} disjunct of $\pfunc{\trans}{r}{b}$.
When the robots are at the final stage (i.e.,
$\readybvar'=\val{2}$) they enable another $\assembleval$ message to
finalise the product and subsequently they reset to their initial conditions.

\section{\ltal: An extension of \ltl  }\label{sec:logic}
We introduce \ltal, an extension of \ltl with the ability to refer and therefore 
	reason about agents interactions.
  We replace the next operator of \ltl with the observation descriptors:
  \emph{possible} $\nxt{O}$ and \emph{necessary} $\alws{O}$, to refer to messages and the intended set of receivers.
  The syntax of formulas $\phi$ and \emph{observation descriptors} $O$ is as follows:
\smallskip

\noindent
$
\begin{array}{@{}l@{\ }r@{\,\,}c@{\,\,}l@{}}
&
\phi &::=& v \mid
  	\neg v \mid
  	\phi \lor \phi \mid
  	\phi \wedge \phi \mid
  	\phi \Until \phi \mid  
  	\phi \Release \phi \mid
  	\nxt{O} \phi \mid
  	\alws{O} \phi\\[2pt]
&
O & ::= &  
\cv \mid \neg \cv \mid \chan \mid \neg \chan \mid  k \mid \neg k \mid 
  	\dat \mid  \neg \dat\mid \exis{O} \mid \all{O}\\[1pt]
	&&& \mid O \lor O 
	  \mid O \wedge O
\end{array}
$
\smallskip

  We use the classic abbreviations $\Impl,\Iff$ and 
  the usual definitions for $\true$ and $\false$. We also introduce the temporal abbreviations $\f\phi\equiv
  \true\Until\phi$ (\emph{eventually}), $\g\phi\equiv\neg\f\neg\phi$
  (globally) and $\varphi\WaitFor\psi\equiv\psi\Release (\psi\vee\varphi)$ (\emph{weak until}). Furthermore we assume that all variables are Boolean because 
  every finite domain can be encoded by multiple Boolean
  variables. For convenience we will, however, use non-Boolean variables when relating to our RMS example. 
  
  The syntax of \ltal is presented in \emph{positive normal form} to facilitate translation into alternating B\"uchi automata (ABW) as shown later. We, therefore, use $\overline{\Theta}$ to denote the dual of formula 
  $\Theta$ where $\Theta$ ranges over either $\phi$ or $O$.
  Intuitively, $\overline{\Theta}$ is obtained from $\Theta$ by switching 
  $\vee$ and $\wedge$ and by applying dual to sub formulas, e.g.,
  $\overline{\phi_1 \Until \phi_2} = \overline{\phi_1} \Release 
  \overline{\phi_2}$, $\overline{\phi_1 \wedge \phi_2} = \overline{\phi_1} \vee 
  \overline{\phi_2},\ $ $\overline{\cv} = \neg \cv$, and $\overline{\exis{O}} = 
  \all{\overline{O}}$.

	Observation descriptors are built from referring to the different parts 
	of the observations and their Boolean combinations. Thus, they refer to the channel in $\schan$, the data 
	variables in $\sdat$, the sender $k$, and the predicate over common variables 
	in $\scv$.
	These predicates are interpreted as sets of possible assignments to common variables, and therefore
	we include existential $\exis{O}$ and universal $\all{O}$ quantifiers over 
	these assignments.
	
The semantics of an observation descriptor $O$ is defined for an observation $m = \tuple{\chan\coma 
\datfun\coma k\coma \pred}$ as follows: 
\smallskip

\noindent
$
\begin{array}{l c l @{\ } | @{\ }l c l}
 m \models \chan'& \text{\bf iff} & \chan = \chan'\quad 
 &\quad 

m \models \neg \chan' & \text{\bf iff} & \chan \neq \chan'\\

 m \models \dat' & \text{\bf iff} &\datfun(\dat') &\quad 

 m \models \neg \dat' & \text{\bf iff} & \neg \datfun(\dat') \\

m \models k' & \text{\bf iff} & k = k'&\quad 
m \models \neg k' & \text{\bf iff} & k \neq k' 
\end{array}\\
\begin{array}{lcc}
			m \models \cv\quad \textbf{iff}\quad \text{for
                          all}\ c \in \pred\ \text{we have}\ 
			c \models \cv&\\

			m \models \neg\cv \ \ \textbf{iff}\ \ \text{there is}\ c \in \pred\ \text{such that}\ 
			c \not\models \cv&\\

			m \models \exis{O} \ \ \textbf{iff}\ \ \text{there is}\ c \in \pred \ \text{such that}\ 

			\tuple{\chan, d, k, \{c\}} \models O&\\

			m \models \all{O} \ \ \textbf{iff}\ \ \text{for all}\ c \in \pred\ \text{it holds that}\ 
			\tuple{\chan, d, k, \{c\}} \models O&\\
		
			m \models O_1 \vee O_2 \quad \textbf{iff}\quad \text{either}\ m \models O_1\ \text{or}\ m 
			\models O_2&\\

			m \models O_1 \wedge O_2 \quad \textbf{iff}\quad  m \models O_1\ \text{and}\ m 
			\models O_2&
	\end{array}
	$
\smallskip

We only comment on the semantics of the descriptors $\exis{O}$ and $\all{O}$ and the rest are standard propositional formulas.
The descriptor $\exis{O}$ requires that at least one
assignment $c$ to the common variables in the sender predicate
$\pred$ satisfies $O$.
Dually $\all{O}$ requires that all assignments in $\pred$ satisfy
$O$.
Using the former, we express properties where we require that the
sender predicate has a possibility to satisfy $O$ while using the
latter we express properties where the sender predicate can only
satisfy $O$.
For instance, both observations $\tuple{\chan, \datfun, k, \cv_1\vee\neg\cv_2}$ 
and $\tuple{\chan, \datfun, k, \cv_1}$
satisfy $\exis{\cv_1}$ while only the latter satisfies
$\all{\cv_1}$.
Furthermore, the observation descriptor $\all{\false}\wedge\chan=\star$ says
that a message is sent on the broadcast channel with a false
predicate.
That is, the message cannot be received by other agents.
In our RMS example in Section~\ref{sec:exp}, the
descriptor $\exis{(\typecvar=\taval)} \wedge 
\all{(\typecvar=\taval)}$ says that the message is intended exactly for
robots of type-$1$.

Note that the semantics of $\exis{O}$ and $\all{O}$ (when nested) ensures that 
the outermost cancels the inner ones, e.g., $\exis{(O_1\vee 
(\all{(\exis{O_2})}))}$ is equivalent to $\exis{(O_1\vee O_2)}$.
Thus, we assume that they are written in the latter normal 
form.

We interpret \ltal formulas over system computations:
\begin{definition}[System computation]~\label{def:syscomp}
A system computation $\rho$ is a function from natural numbers $N$ to
$\Exp{\sysvar}\times M$ where $\sysvar$ is the set of state variable
propositions and $M=\schan\times 2^{\sdat}\times K\times\dexp{\scv}$
is the set of possible observations. That is, $\rho$ includes values
for the variables in $\Exp{\sysvar}$ and an observation in $M$ at each
time instant.
\end{definition}

We denote by $\sstate{\id}$ the system state at the $i$-th 
time point of the system computation.
Moreover, we denote the suffix of $\tracevar{}$ starting with the $i$-th
state by $\tracevar{\geq i}$ and we use $m_{i}$ to denote the observation
$\tuple{\chan, \datfun, k, \pred}$ in $\tracevar{}$ at time point
$i$.

The semantics of an \ltal formula $\varphi$ is defined for a computation 
$\tracevar{}$ at a time point $\id$ as follows:\smallskip

	$\begin{array}{lcc}		
			\tracevar{\geq i}\models v\ \ \textbf{iff}\ \ \sstate{\id} \models v\quad
		\text{and} \quad
			 \tracevar{\geq i}\models \neg v\ \ \textbf{iff}\ \  \sstate{\id} 
			 \not\models v;\\		
	
			\tracevar{\geq i}\models \phi_2\vee\phi_2\ \ \textbf{iff}\ \ \tracevar{\geq 
			i}\models\phi_1\ \text{ or }\ \tracevar{\geq 
			i}\models\phi_2;\\

			\tracevar{\geq i}\models \phi_2\wedge\phi_2\ \ \textbf{iff}\ \ 
			\tracevar{\geq i}\models\phi_1\ \text{ and }\ \tracevar{\geq 
			i}\models\phi_2;\\

			\tracevar{\geq i}\models \phi_1\Until\phi_2\ \ \textbf{iff}\ \text{there exists}\ 
			j\geq i\ \text{s.t.}\ \tracevar{ \geq j} \models \phi_2\ \text{and,}\\ \hfill\text{for 
			every}\ i\leq k<j,\ \tracevar{\geq k}\models\phi_1;\qquad\\

			\tracevar{\geq i}\models \phi_1\Release\phi_2\ \ \textbf{iff}\ \text{for every}\ j 
			\geq i\ \text{either}\ \tracevar{\geq j} \models \phi_2\ \text{or}\\ \hfill\text{there exists}\ i \leq 
			k<j\ \text{s.t.}\ \tracevar{\geq k}\models\phi_1;\\
	
			\tracevar{\geq i}\models\nxt{O}\phi\ \ \textbf{iff}\ \ m_i\models O\ 
			\text{ and }\  \tracevar{\geq i+1}\models\phi;\\			
	
			\tracevar{\geq i}\models\alws{O}\phi\ \ \textbf{iff}\ m_i\models O\ 
			\text{ implies }\  \tracevar{\geq i+1}\models\phi.
\end{array}$\smallskip

Intuitively, the temporal formula $\nxt{O}\phi$ is satisfied on the computation 
$\tracevar{}$ at point $\id$ if the observation $m_i$ satisfies $O$ \emph{and}
$\phi$ is satisfied on the suffix computation $\tracevar{\geq \id+1}$.
On the other hand, the formula $\alws{O}\phi$ is satisfied on the computation 
$\tracevar{}$ at point $\id$ if $m_i$ satisfying $O$ \emph{implies}
that $\phi$ is satisfied on the suffix computation $\tracevar{\geq \id+1}$.
Other formulas are interpreted exactly as in classic \ltl.

With observation descriptors we can refer to the intention of
senders.
For example,
$O_1:=\exis{(\typecvar=\taval)}\wedge\exis{(\typecvar=\tbval)}\wedge
\all{(\typecvar=\taval \vee \typecvar=\tbval)}$ specifies that the target of the
message is \qt{exactly and only} both type-1 and type-2 robots.
Thus, we may specify that whenever the line agent \qt{$l$} recruits for a product with identity 1, 
 it notifies both type-1 and type-2 robots:
$\g((\prdlvar=\val{1} \wedge \stlvar=\pendval \wedge \nxt{l\wedge \chan=\toall}\true)
\rightarrow \nxt{O_1}\true)$. Namely, whenever the line agent is in the
pending state and tasked with product $1$ it notifies both type-1 and
type-2 robots by a broadcast. 
The pattern ``After $q$ have exactly two $p$ until $r$''
\cite{MR15,DAC99} can be easily expressed in \ltl and can be used to
check the formation protocol. 
Consider the following formulas.
Let $$\varphi_1:=\langle \msgdvar=\teamval\wedge \nodvar=\val{2} \wedge
  \exis{(\typecvar=\taval)}\rangle\true$$, i.e., a team message is sent to
type-1 robots requiring two robots.
Let $$\varphi_2:=\nxt{\msgdvar=\formval \wedge \exis{(\typecvar=\taval
    )}}\true$$, i.e., a formation message is sent to type-1 robots.
Let $\varphi_3:=\nxt{\chan=A}\true$, i.e., a message is sent on
channel $A$.
Then,
``After $\varphi_1$ have exactly two $\varphi_2$ until $\varphi_3$''
says that whenever a team message is sent to robots of type-1
requiring two robots, then two form messages destined for type-1
robots will follow before using the multicast channel.
That is, two type-1 robots join the team before a (blocking) multicast
on channel $A$ may become possible.

We can also reason at a
local rather than a global level.
For instance, we can specify that robots follow a ``correct''
utilisation of channel $A$. Formally,
let $O_1(t):=\msgdvar{=}\teamval \wedge
  \exis(\typecvar{=}\val{t})$, i.e., a team message is sent
  to robots of type $\val{t}$;
  $O_2(k,t):=\msgdvar{=}\formval \wedge
  \neg k \wedge
  \nodvar{=}\val{0} \wedge \exis{(\neg \assigncvar \wedge
    \typecvar{=}\val{t})}$, i.e., a robot different from $k$ sends a
  form message saying no more robots are needed and this message is
  sent to unassigned type $\val{t}$ robots; and let
  $O_3(t):=\msgdvar{=}\assembleval \wedge \chan=A \wedge
    \readycvar{=}\val{2} \wedge \exis{(\typecvar{=}\val{t})}$, i.e., an
    assembly message is sent on channel $A$ to robots of type
    $\val{t}$ who reached stage 2 of the production.
    Thus, for robot $k$ of type $\val{t}$, the formulas
    
    \smallskip
    \noindent
    $ \small
    \begin{array}{l@{\ \ }l}
    (1) & \varphi_1(t):=(\lnklvar{\neq} A) \WaitFor \nxt{O_1(t)}\true\\
      (2)&\varphi_2(k,t):=\g(\alws{O_2(k,t)\vee O_3(t)}\varphi_1(t))
    \end{array}
    $
    
    \smallskip
\noindent
    state that: (1) robots are not connected to channel $A$ until they get
a team message, inviting them to join a team; (2) if either they
are not selected ($O_2(k,t)$) or they finished production after selection
 ($O_3(t)$) then they disconnect again until the next team message.
This reduces to checking the  ``correct'' utilisation of channel $A$ to individual
level, by verifying these properties on all types of robots
independently. 
%
%
By allowing the logic to relate to the set of targeted robots, verifying all
 targeted robots separately entails the
correct \qt{group usage} of channel $A$.


%
\comment{
\begin{figure*}
$\small
\begin{array}{c}
\band{k}{}{
\tuple{
\begin{array}{c}
\msf{prd}_k=1
\end{array}
}\to  
\tuple{(\anglebr{k\wedge \msf{\chan=lnk}_k}\false)\ \Until 
\bnxt{
\begin{array}{l c l l}
\msf{d(\rulename{msg})=team} \wedge k\wedge \msf{\chan=\star\wedge 
d(\rulename{no})=2}\wedge\\
\exis{(\msf{cv_1=\rulename{t1}\vee cv_1=\rulename{t2}})}\wedge
{\msf{d(\rulename{lnk})=lnk}_k}\\
\end{array}
}\true}}\qquad\qquad\qquad\qquad\qquad\\[4ex]

\begin{array}{cc}
\wedge\band{k'\neq k}{}{}{\tuple{
\begin{array}{c}
\msf{\neg asgn}_{k'}\wedge \\ 
(\msf{role}_{k'}=\rulename{t1}\vee\\ 
\msf{role}_{k'}=\rulename{t2})
\end{array}}
}\to
\tuple{
\anglebr{
\begin{array}{l c l l}
\msf{d(\rulename{msg})=team} \wedge
 k\wedge\\ {\msf{d(\rulename{lnk})=lnk}_k}\wedge \msf{\chan=\star}\wedge\\
\exis{(\msf{cv_1=\rulename{t1}\vee cv_1=\rulename{t2}})}
\end{array}
}
\tuple{
\tuple{\msf{lnk}_{k'}=\msf{d(lnk)\wedge\msf{ready}_{k'}}}
 \Release\  
{
\anglebr{\begin{array}{c}
\msf{d(\rulename{msg})=assemble}\wedge\\ k\wedge 
\msf{\chan=lnk}_k\end{array}}\false
}}}
\end{array}
\end{array}
$

\end{figure*}
}
%
%

To model-check \ltal formulas we first need to translate them to B\"uchi 
automata. The following theorem states that the set of computations satisfying a
given formula are exactly the ones accepted by some finite automaton
on infinite words. Before we proceed with the theorem, we fix the following notations.
Given a tuple $T = \tuple{t_1, \dots, t_n}$, we use $T[i]$ to denote the 
element $t_i$ of $T$ and $T[x/i]$ to denote the tuple where $t_i$ is replaced 
with $x$. \vspace{-.6mm}

\begin{theorem} \label{thm:main}
For every \ltal formula $\phi$, there is an ABW $A_{\phi}=\conf{Q, \Sigma, M, \transmain, q_0,  F\subseteq Q }$ such that $\lang{A_{\phi}}$ is exactly the set of computations satisfying the formula $\phi$.
\end{theorem}\vspace{-3.53mm}
\begin{proof} 
The set of states $Q$ is the set of all sub formulas of $\phi$ with $\phi$ 
being the initial state $q_0$.
The automaton has two alphabets, namely a \emph{state-alphabet} $\Sigma = 
2^{\sysvar}$ and a \emph{message-alphabet} $M = \schan \times 2^{\sdat} \times 
K \times \dexp{\scv}$.
The set $F$ of accepting states consists of all sub formulas of the form 
$\phi_1 \Release \phi_2$.
The transition relation $\transmain: Q \times \Sigma \times M \rightarrow 
\B^{+}(Q)$ is defined inductively on the structure of $\phi$.
It also relies on an auxiliary function $\transaux: O \times M \rightarrow 
\bb$ to evaluate observations and is defined recursively on $O$ as follows:

\begin{itemize}[label={$\bullet$}, topsep=0pt, itemsep=0pt, leftmargin=10pt]
	\item
		$\transaux(\cv, m) = \band{}{}{\hspace{-.5mm}_{c \in m[4]}\ c(\cv)}\ $  and\vspace{-2mm}   \\
		$\transaux(\neg \cv, m)=\bor{}{}{\hspace{-1mm}_{c \in m[4]}\ \neg c(\cv)}$;\vspace{-1mm}

	\item
		$\transaux(\chan, m) = \true$ if $m[1] = \chan$ and $\false$ otherwise;

	\item
		$\transaux(\neg \chan, m) = \true$ if $m[1]\neq\chan$ and $\false$ 
		otherwise;

	\item
		$\transaux(\dat, m) = \true$ if $m[2](\dat)$ and $\false$ otherwise;

	\item
		$\transaux(\neg \dat, m) = \true$ if $\neg m[2](\dat)$ and $\false$ 
		otherwise;

	\item
		$\transaux(k, m) = \true$ if $m[3] = k$ and $\false$ otherwise;

	\item
		$\transaux(\neg k, m) = \true$ if $m[3]\neq k$ and $\false$ otherwise;

	\item
		$\transaux(O_1 \vee O_2, m) = \transaux(O_1, m) \vee \transaux(O_2, m)$ 
		and  \\
		$\transaux(O_1 \wedge O_2, m) =\transaux(O_1, m) \wedge \transaux(O_2, m)$;

	\item
		$\transaux(\exis{O}, m) = \bor{}{}{\hspace{-1mm}_{c \in m[4]}\ \transaux(O, m[c/4])}$;\vspace{-2mm}

	\item 
		 $\transaux(\all{O},m)=\band{}{}{\hspace{-.5mm}_{c\in m[4]}\ \transaux(O,m[c/4])}$.
\end{itemize}

The transition relation $\transmain$ is defined as follows:

\begin{itemize}[label={$\bullet$}, topsep=0pt, itemsep=0pt, leftmargin=10pt]
	\item
		$\transmain(v, \sigma, m) = \true$ if $v\in\sigma$ and $\false$ otherwise;

	\item
		$\transmain(\neg v, \sigma, m) = \true$ if $ v \not \in \sigma$ and 
		$\false$ otherwise;

	\item
		$\transmain(\phi_1 \vee \phi_2, \sigma, m) = \transmain(\phi_1, \sigma, m) 
		\vee \transmain(\phi_2, \sigma, m)$;

	\item
		$\transmain(\phi_1 \wedge \phi_2, \sigma, m) = \transmain(\phi_1, \sigma, 
		m) \wedge \transmain(\phi_2, \sigma, m)$;

	\item
		$\transmain(\phi_1 \Until \phi_2, \sigma, m) = \transmain(\phi_1, \sigma, 
		m) \wedge\phi_1 \Until \phi_2 \vee \transmain(\phi_2, \sigma, m)$;

	\item
		$\transmain(\phi_1 \Release \phi_2, \sigma, m) = (\transmain(\phi_1, 
		\sigma, m) \vee \phi_1 \Release \phi_2) \wedge \transmain(\phi_2, \sigma, 
		m)$;
	
	\item
		$\transmain(\nxt{O}\phi, \sigma, m) = \phi \wedge \transaux(O, m)$;
		
	\item
		$\transmain(\alws{O}\phi, \sigma, m) = \phi \vee \transaux(\overline{O}, 
		m)$.
\end{itemize}

The proof of correctness of this construction proceeds by induction on the structure of $\phi$.
\end{proof}
Observe that, from Theorem~\ref{thm:main}, the number of states in $A_{\phi}$ 
is linear in the size of $\phi$, i.e., $\size{Q}$ is in $\bigo{\size{\phi}}$. 
The size of the transition relation $\size{\transmain}$ is in $\bigo{\size{Q}^2.\size{\Sigma}.\size{M}}$, i.e., $\size{\transmain}$ is in
${\size{\phi}^2.\size{\schan}.\size{K}.\Exp{\bigo{\size{\sysvar}+\size{\sdat}+\Exp{\size{\scv}}}}}$. Furthermore the
evaluation function $\transaux$ can be computed in
$\bigo{\size{O}.\size{\scv}}$ time and in
$\bigo{\log{\size{O}}+\log{\size{\scv}}}$ space. Finally, the size of
the alternating automaton $\size{A_{\phi}}$ is in $\bigo{\size{Q}.\size{\transmain}}$, i.e., $\size{A_{\phi}}$ is in $(\size{\phi})^3.\size{\schan}.\size{K}.\Exp{\bigo{\size{\sysvar}+\size{\sdat}+\Exp{\size{\scv}}}}$.

\begin{proposition}[\cite{mh84}]\label{prop:aton}
For every alternating B\"uchi automaton $A$ there is a nondeterministic B\"uchi automaton $A'$ such that $\lang{A'}=\lang{A}$ and $\size{Q'}$ is in $\Exp{\bigo{\size{Q}}}$.
\end{proposition}

By Theorem~\ref{thm:main} and Proposition~\ref{prop:aton}, we have that:
\begin{corollary} For every formula $\phi$ there is a B\"uchi automaton $A$ with a state-alphabet $\Sigma=2^{\sysvar}$ and a message-alphabet 
$M=\schan\times 2^{\sdat}\times K\times\dexp{\scv} $ where $A=\langle Q, \Sigma,$  $M, S^0, \delta, F\rangle$ and $\lang{A}$ is exactly the set of computations satisfying the formula $\phi$ such that:

\begin{itemize}[label={$\bullet$}, topsep=0pt, itemsep=0pt, leftmargin=10pt]
\item $\size{Q}$ is in $\Exp{\bigo{\size{\phi}}}$
and 
  $\size{\delta}$ is in $\bigo{\size{Q}^2.\size{\Sigma}.\size{M}}$, i.e., $\size{\delta}$ is in ${\size{\schan}.\size{K}.\Exp{\bigo{\size{\phi}+\size{\sysvar}+\size{\sdat}+\Exp{\size{\scv}}}}}$

\item The space requirements for building the B\"uchi automaton on-the-fly is $\nlog{(\size{Q}.\size{\delta})}$, i.e., it is in\\ $\bigo{\log{\size{\schan}}+\log{\size{K}}+\size{\phi}+\size{\sysvar}+\size{\sdat}+\Exp{\size{\scv}}}$

\item The size of the B\"uchi automaton is $\size{Q}.\size{\delta}$, i.e., $\size{A}$ is in $\size{\schan}.\size{K}.\Exp{\bigo{\size{\sysvar}+\size{\sdat}+\Exp{\size{\scv}}}}$
\end{itemize}
\label{cor:1}
\end{corollary}

We do not see the double exponential in the
automaton size with respect to common variables $\scv$ as a major
restriction.
Note that $|\scv|$ does not grow with respect
to the size of the formula or to the number of the agents.
Thus, if we limit the number of common variables (which
should be small), efficient verification can still be attainable.  

\begin{theorem}\label{thm:satisfiability}
The satisfiability problem of \ltal is \pspace-complete with respect to 
$\size{\phi}$, $\size{\sysvar},\ \size{\sdat},\ \log{\size{\schan}},\ 
\log{\size{K}}$ and \expspace with respect to $\size{\scv}$. 
\end{theorem}

\begin{theorem}\label{thm:modcheck}
The model-checking problem of \ltal is \pspace-complete with respect to $\size{Sys}$, $\size{\phi}$, $\size{\sysvar},\ \size{\sdat},\ \log{\size{\schan}},\\ \log{\size{K}}$ and \expspace with respect to $\size{\scv}$. 
\end{theorem}

\newcommand{\lowerbound}{

We now investigate the hardness of \ltal satisfiability and model-checking by showing a reduction from EXPSPACE tiling problem~\cite{Wang90,Boas97} to satisfiability and model-checking, providing EXPSPACE hardness. 

\begin{theorem}[hardness]
	\label{thm:hardness}
	The satisfiability and model-checking problems for properties in \ltal are EXPSPACE-hard.
\end{theorem}

\begin{proof}
  The proof is in Appendix~\ref{sec:appendix}.
  Let $\Exp{m}$ be the width of the required tiling.
  The crux of the proof is in checking that two tiles that are
  $\Exp{m}$ apart from each other satisfy some consistency rules.
  We use common variables that can describes tiles and a number
  counting to at most $\Exp{m}-1$.
  A predicate over the common variables has a set of possible truth
  assignments to these variables.
  Thus, we can think about the predicate as a relation between
  locations (modulo $\Exp{m}$) and tiles.
  We use the observation descriptors to force the predicates to
  specify that this relation is actually a function and so that the
  function memorizes a full row in the tiling (i.e., $\Exp{m}$ tiles).
  Then, by looking on the predicate we can compare the current tile to
  a tile that appeared $\Exp{m}$ locations ago. 
\end{proof}
}


\section{Concluding Remarks}\label{sec:conc}
We introduced a formalism that combines message-passing and shared-memory to facilitate realistic modelling of distributed MAS. 
A system is defined as a set of distributed agents  
 that execute concurrently and only interact by message-passing. Each agent controls its local behaviour as in Reactive Modules~\cite{AH99b,FHNPSV11} while 
interacting externally by message passing as in 
$\pi$-calculus-like formalisms~\cite{pi1,info19}. Thus, we decouple the individual behaviour of an agent from its external interactions to facilitate reasoning about either one separately. We also make it easy to model interaction features of MAS, that may only be tediously hard-coded in existing formalisms.  

We introduced an extension to \ltl, named \ltal, that characterises messages and their targets. This way we may not only be able to reason about the intentions of agents in communication, but also we may explicitly specify their interaction protocols. Finally, we showed that the model-checking problem for \ltal is \expspace only 
with respect to the number of common variables and \pspace-complete with respect 
to the rest of the input.\smallskip

\noindent
{\bf Related works.}
As mentioned before, formal modelling is highly influenced by traditional
formalisms used for verification, see~\cite{AH99b,FHMV95}.
These formalisms are, however, very abstract in that their models
representations are very close to their mathematical interpretations
(i.e., the underlying transition systems).
Although this  may make it easy to conduct some logical
analysis~\cite{AHK02,CHP10,MMPV14} on models, it does imply that most of the
high-level MAS features may only be hard-coded, and thus leading to
very detailed models that may not be tractable or efficiently
implementable.
This concern has been already recognised and thus more formalisms have
been proposed, e.g., \emph{Interpreted Systems Programming Language}
(ISPL)~\cite{LQR17} and MOCHA~\cite{AHMQRT98} are proposed as
implementation languages of \emph{Interpreted Systems}
(IS)~\cite{FHMV95} and \emph{Reactive Modules} (RM)~\cite{AH99b} 
respectively.
They are still either fully synchronous or shared-memory based and thus do not support
flexible coordination and/or interaction interfaces.
A recent attempt to add dynamicity in this sense has been adopted by
\emph{visibly CGS} (vCGS)~\cite{BBDM19}: an extension of \emph{Concurrent-Game 
Structures} (CGS)~\cite{AHK02} to enable agents to dynamically hide/reveal 
their internal states to selected agents.
However, vCGS relies on an assumption of \cite{stop} which requires that agents 
know the identities of each other.
This, however, only works for closed systems with a fixed number of agents. 

Other attempts to add dynamicity and reconfiguration 
include dynamic I/O automata \cite{DBLP:journals/iandc/AttieL16},
Dynamic reactive modules of Alur and Grosu \cite{AG04},
Dynamic reactive modules of Fisher et
al.~\cite{FHNPSV11}, and open MAS~\cite{KLPP19}.
However, their main interest was in supporting dynamic creation of agents.
Thus, the reconfiguration of communication was not their main
interest. 
While \rcp may be easily extended to support dynamic creation of
agents, none of these formalisms may easily be used
to control the targets of communication and dissemination of
information.

As for logics we differ from traditional languages like \ltl and \ctl 
in that our formula may refer to messages and their constraints.
This is, however, different from the atomic labels of {\pdl}~\cite{FL79} and 
{modal $\mu$-calculus}~\cite{Koz83} in that \ltal mounts complex and structured 
observations on which designers may 
predicate on.
Thus the interpretation of a formula includes information about the causes of 
variable assignments and the interaction protocols among agents.
Such extra information may prove useful in developing compositional 
verification techniques.

\smallskip

\noindent
{\bf Future works.}  
We plan to provide tool support for \rcp and \ltal, but with a more 
user-friendly syntax.

We want to exploit the interaction mechanisms in \rcp and the extra
information in \ltal formulas to conduct verification
compositionally.
As mentioned, we believe that relating to sender intentions will
facilitate that. 

We intend to study the relation with respect to temporal epistemic logic 
\cite{HV89}.
Although we do not provide explicit knowledge operators, the combination of 
data exchange, receivers selection, and enabling/disabling of synchronisation 
allow agents to dynamically deduce information about each other.
Furthermore we want to extend \rcp to enable dynamic creation of agents while 
reconfiguring communication.

Finally, we want to target
the distributed synthesis problem~\cite{FS05}.
Several fragments of the problem have been proven to be decidable, e.g.,
when the information of agents is arranged 
hierarchically~\cite{BMMRV17}, the number of agents is limited~\cite{GPW18}, or 
the actions are made public~\cite{BLMR17a}.
We conjecture that the ability to disseminate information and reason about it might prove 
useful in this setting.


\bibliographystyle{ACM-Reference-Format}  
\bibliography{biblio}   


\begin{thebibliography}{00}


\ifx \showCODEN    \undefined \def \showCODEN     #1{\unskip}     \fi
\ifx \showDOI      \undefined \def \showDOI       #1{#1}\fi
\ifx \showISBNx    \undefined \def \showISBNx     #1{\unskip}     \fi
\ifx \showISBNxiii \undefined \def \showISBNxiii  #1{\unskip}     \fi
\ifx \showISSN     \undefined \def \showISSN      #1{\unskip}     \fi
\ifx \showLCCN     \undefined \def \showLCCN      #1{\unskip}     \fi
\ifx \shownote     \undefined \def \shownote      #1{#1}          \fi
\ifx \showarticletitle \undefined \def \showarticletitle #1{#1}   \fi
\ifx \showURL      \undefined \def \showURL       {\relax}        \fi
\providecommand\bibfield[2]{#2}
\providecommand\bibinfo[2]{#2}
\providecommand\natexlab[1]{#1}
\providecommand\showeprint[2][]{arXiv:#2}

\bibitem[\protect\citeauthoryear{{Abd Alrahman}, {De Nicola}, and Loreti}{{Abd
  Alrahman} et~al\mbox{.}}{2019}]%
        {info19}
\bibfield{author}{\bibinfo{person}{Yehia {Abd Alrahman}},
  \bibinfo{person}{Rocco {De Nicola}}, {and} \bibinfo{person}{Michele Loreti}.}
  \bibinfo{year}{2019}\natexlab{}.
\newblock \showarticletitle{A calculus for collective-adaptive systems and its
  behavioural theory}.
\newblock \bibinfo{journal}{{\em Inf. Comput.\/}}  \bibinfo{volume}{268}
  (\bibinfo{year}{2019}).
\newblock
\showDOI{%
\url{https://doi.org/10.1016/j.ic.2019.104457}}


\bibitem[\protect\citeauthoryear{Aguilera, Ben-David, Calciu, Guerraoui,
  Petrank, and Toueg}{Aguilera et~al\mbox{.}}{2018}]%
        {mm}
\bibfield{author}{\bibinfo{person}{Marcos~K. Aguilera}, \bibinfo{person}{Naama
  Ben-David}, \bibinfo{person}{Irina Calciu}, \bibinfo{person}{Rachid
  Guerraoui}, \bibinfo{person}{Erez Petrank}, {and} \bibinfo{person}{Sam
  Toueg}.} \bibinfo{year}{2018}\natexlab{}.
\newblock \showarticletitle{Passing Messages While Sharing Memory}. In
  \bibinfo{booktitle}{{\em Proceedings of the 2018 ACM Symposium on Principles
  of Distributed Computing}} {\em (\bibinfo{series}{PODC '18})}.
  \bibinfo{publisher}{ACM}, \bibinfo{address}{New York, NY, USA},
  \bibinfo{pages}{51--60}.
\newblock
\showISBNx{978-1-4503-5795-1}
\showDOI{%
\url{https://doi.org/10.1145/3212734.3212741}}


\bibitem[\protect\citeauthoryear{Aguilera, Ben-David, Guerraoui, Marathe, and
  Zablotchi}{Aguilera et~al\mbox{.}}{2019}]%
        {rdma}
\bibfield{author}{\bibinfo{person}{Marcos~K. Aguilera}, \bibinfo{person}{Naama
  Ben-David}, \bibinfo{person}{Rachid Guerraoui}, \bibinfo{person}{Virendra
  Marathe}, {and} \bibinfo{person}{Igor Zablotchi}.}
  \bibinfo{year}{2019}\natexlab{}.
\newblock \showarticletitle{The Impact of RDMA on Agreement}. In
  \bibinfo{booktitle}{{\em Proceedings of the 2019 ACM Symposium on Principles
  of Distributed Computing}} {\em (\bibinfo{series}{PODC '19})}.
  \bibinfo{publisher}{ACM}, \bibinfo{address}{New York, NY, USA},
  \bibinfo{pages}{409--418}.
\newblock
\showISBNx{978-1-4503-6217-7}
\showDOI{%
\url{https://doi.org/10.1145/3293611.3331601}}


\bibitem[\protect\citeauthoryear{Alur and Grosu}{Alur and Grosu}{2004}]%
        {AG04}
\bibfield{author}{\bibinfo{person}{Rajeev Alur} {and} \bibinfo{person}{Radu
  Grosu}.} \bibinfo{year}{2004}\natexlab{}.
\newblock \bibinfo{booktitle}{{\em Dynamic {R}eactive {M}odules}}.
\newblock \bibinfo{type}{{T}echnical {R}eport} 2004/6.
  \bibinfo{institution}{Stony Brook}.
\newblock


\bibitem[\protect\citeauthoryear{Alur and Henzinger}{Alur and
  Henzinger}{1999}]%
        {AH99b}
\bibfield{author}{\bibinfo{person}{Rajeev Alur} {and} \bibinfo{person}{Thomas
  Henzinger}.} \bibinfo{year}{1999}\natexlab{}.
\newblock \showarticletitle{{Reactive Modules}}.
\newblock \bibinfo{journal}{{\em Formal Methods in System Design\/}}
  \bibinfo{volume}{15}, \bibinfo{number}{1} (\bibinfo{year}{1999}),
  \bibinfo{pages}{7--48}.
\newblock


\bibitem[\protect\citeauthoryear{Alur, Henzinger, and Kupferman}{Alur
  et~al\mbox{.}}{2002}]%
        {AHK02}
\bibfield{author}{\bibinfo{person}{Rajeev Alur}, \bibinfo{person}{{Thomas A.}
  Henzinger}, {and} \bibinfo{person}{Orna Kupferman}.}
  \bibinfo{year}{2002}\natexlab{}.
\newblock \showarticletitle{Alternating-Time Temporal Logic.}
\newblock \bibinfo{journal}{{\em J. ACM\/}} \bibinfo{volume}{49},
  \bibinfo{number}{5} (\bibinfo{year}{2002}), \bibinfo{pages}{672--713}.
\newblock
\showISBNx{3-540-65493-3}
\showISSN{0004-5411}
\showDOI{%
\url{https://doi.org/10.1145/585265.585270}}


\bibitem[\protect\citeauthoryear{Alur, Henzinger, Mang, Qadeer, Rajamani, and
  Tasiran}{Alur et~al\mbox{.}}{1998}]%
        {AHMQRT98}
\bibfield{author}{\bibinfo{person}{Rajeev Alur}, \bibinfo{person}{{Thomas A.}
  Henzinger}, \bibinfo{person}{{Freddy Y. C.} Mang}, \bibinfo{person}{Shaz
  Qadeer}, \bibinfo{person}{{Sriram K.} Rajamani}, {and}
  \bibinfo{person}{Serdar Tasiran}.} \bibinfo{year}{1998}\natexlab{}.
\newblock \showarticletitle{{{MOCHA:} Modularity in Model Checking.}}. In
  \bibinfo{booktitle}{{\em {CAV}'98}}. \bibinfo{pages}{521--525}.
\newblock


\bibitem[\protect\citeauthoryear{Attiya, Bar-Noy, and Dolev}{Attiya
  et~al\mbox{.}}{1995}]%
        {stop}
\bibfield{author}{\bibinfo{person}{Hagit Attiya}, \bibinfo{person}{Amotz
  Bar-Noy}, {and} \bibinfo{person}{Danny Dolev}.}
  \bibinfo{year}{1995}\natexlab{}.
\newblock \showarticletitle{Sharing Memory Robustly in Message-passing
  Systems}.
\newblock \bibinfo{journal}{{\em J. ACM\/}} \bibinfo{volume}{42},
  \bibinfo{number}{1} (\bibinfo{date}{Jan.} \bibinfo{year}{1995}),
  \bibinfo{pages}{124--142}.
\newblock
\showISSN{0004-5411}
\showDOI{%
\url{https://doi.org/10.1145/200836.200869}}


\bibitem[\protect\citeauthoryear{Baumann, Barham, Dagand, Harris, Isaacs,
  Peter, Roscoe, Sch\"{u}pbach, and Singhania}{Baumann et~al\mbox{.}}{2009}]%
        {mno}
\bibfield{author}{\bibinfo{person}{Andrew Baumann}, \bibinfo{person}{Paul
  Barham}, \bibinfo{person}{Pierre-Evariste Dagand}, \bibinfo{person}{Tim
  Harris}, \bibinfo{person}{Rebecca Isaacs}, \bibinfo{person}{Simon Peter},
  \bibinfo{person}{Timothy Roscoe}, \bibinfo{person}{Adrian Sch\"{u}pbach},
  {and} \bibinfo{person}{Akhilesh Singhania}.} \bibinfo{year}{2009}\natexlab{}.
\newblock \showarticletitle{The Multikernel: A New OS Architecture for Scalable
  Multicore Systems}. In \bibinfo{booktitle}{{\em Proceedings of the ACM SIGOPS
  22Nd Symposium on Operating Systems Principles}} {\em (\bibinfo{series}{SOSP
  '09})}. \bibinfo{publisher}{ACM}, \bibinfo{address}{New York, NY, USA},
  \bibinfo{pages}{29--44}.
\newblock
\showISBNx{978-1-60558-752-3}
\showDOI{%
\url{https://doi.org/10.1145/1629575.1629579}}


\bibitem[\protect\citeauthoryear{Belardinelli, Boureanu, Dima, and
  Malvone}{Belardinelli et~al\mbox{.}}{2019}]%
        {BBDM19}
\bibfield{author}{\bibinfo{person}{Francesco Belardinelli},
  \bibinfo{person}{Ioana Boureanu}, \bibinfo{person}{Catalin Dima}, {and}
  \bibinfo{person}{Vadim Malvone}.} \bibinfo{year}{2019}\natexlab{}.
\newblock \showarticletitle{Verifying Strategic Abilities in Multi-agent
  Systems with Private Data-Sharing}. In \bibinfo{booktitle}{{\em Proceedings
  of the 18th International Conference on Autonomous Agents and MultiAgent
  Systems, {AAMAS} '19, Montreal, QC, Canada, May 13-17, 2019}}.
  \bibinfo{pages}{1820--1822}.
\newblock


\bibitem[\protect\citeauthoryear{Belardinelli, Lomuscio, Murano, and
  Rubin}{Belardinelli et~al\mbox{.}}{2017}]%
        {BLMR17a}
\bibfield{author}{\bibinfo{person}{Francesco Belardinelli},
  \bibinfo{person}{Alessio Lomuscio}, \bibinfo{person}{Aniello Murano}, {and}
  \bibinfo{person}{Sasha Rubin}.} \bibinfo{year}{2017}\natexlab{}.
\newblock \showarticletitle{Verification of Multi-agent Systems with Imperfect
  Information and Public Actions}. In \bibinfo{booktitle}{{\em Proceedings of
  the 16th Conference on Autonomous Agents and MultiAgent Systems, {AAMAS}
  2017, S{\~{a}}o Paulo, Brazil, May 8-12, 2017}}. \bibinfo{pages}{1268--1276}.
\newblock


\bibitem[\protect\citeauthoryear{Berthon, Maubert, Murano, Rubin, and
  Vardi}{Berthon et~al\mbox{.}}{2017}]%
        {BMMRV17}
\bibfield{author}{\bibinfo{person}{Rapha{\"{e}}l Berthon},
  \bibinfo{person}{Bastien Maubert}, \bibinfo{person}{Aniello Murano},
  \bibinfo{person}{Sasha Rubin}, {and} \bibinfo{person}{{Moshe Y.} Vardi}.}
  \bibinfo{year}{2017}\natexlab{}.
\newblock \showarticletitle{Strategy logic with imperfect information}. In
  \bibinfo{booktitle}{{\em 32nd Annual {ACM/IEEE} Symposium on Logic in
  Computer Science, {LICS} 2017, Reykjavik, Iceland, June 20-23, 2017}}.
  \bibinfo{pages}{1--12}.
\newblock
\showDOI{%
\url{https://doi.org/10.1109/LICS.2017.8005136}}


\bibitem[\protect\citeauthoryear{{C. Attie} and {A. Lynch}}{{C. Attie} and {A.
  Lynch}}{2016}]%
        {DBLP:journals/iandc/AttieL16}
\bibfield{author}{\bibinfo{person}{Paul {C. Attie}} {and}
  \bibinfo{person}{Nancy {A. Lynch}}.} \bibinfo{year}{2016}\natexlab{}.
\newblock \showarticletitle{Dynamic input/output automata: {A} formal and
  compositional model for dynamic systems}.
\newblock \bibinfo{journal}{{\em Inf. Comput.\/}}  \bibinfo{volume}{249}
  (\bibinfo{year}{2016}), \bibinfo{pages}{28--75}.
\newblock
\showDOI{%
\url{https://doi.org/10.1016/j.ic.2016.03.008}}


\bibitem[\protect\citeauthoryear{Chatterjee, Henzinger, and
  Piterman}{Chatterjee et~al\mbox{.}}{2010}]%
        {CHP10}
\bibfield{author}{\bibinfo{person}{Krishnendu Chatterjee},
  \bibinfo{person}{{Thomas A.} Henzinger}, {and} \bibinfo{person}{Nir
  Piterman}.} \bibinfo{year}{2010}\natexlab{}.
\newblock \showarticletitle{Strategy logic}.
\newblock \bibinfo{journal}{{\em Inf. Comput.\/}} \bibinfo{volume}{208},
  \bibinfo{number}{6} (\bibinfo{date}{jun} \bibinfo{year}{2010}),
  \bibinfo{pages}{677--693}.
\newblock
\showISSN{08905401}
\showDOI{%
\url{https://doi.org/10.1016/j.ic.2009.07.004}}


\bibitem[\protect\citeauthoryear{Dorigo}{Dorigo}{2014}]%
        {swarm}
\bibfield{author}{\bibinfo{person}{Marco Dorigo}.}
  \bibinfo{year}{2014}\natexlab{}.
\newblock \showarticletitle{The Swarm-bots and swarmanoid experiments in swarm
  robotics}. In \bibinfo{booktitle}{{\em 2014 {IEEE} International Conference
  on Autonomous Robot Systems and Competitions, {ICARSC} 2014, Espinho,
  Portugal, May 14-15, 2014}}. \bibinfo{pages}{1}.
\newblock
\showDOI{%
\url{https://doi.org/10.1109/ICARSC.2014.6849753}}


\bibitem[\protect\citeauthoryear{Dwyer, Avrunin, and Corbett}{Dwyer
  et~al\mbox{.}}{1999}]%
        {DAC99}
\bibfield{author}{\bibinfo{person}{{Matthew B.} Dwyer},
  \bibinfo{person}{{George S.} Avrunin}, {and} \bibinfo{person}{{James C.}
  Corbett}.} \bibinfo{year}{1999}\natexlab{}.
\newblock \showarticletitle{Patterns in Property Specifications for
  Finite-State Verification}. In \bibinfo{booktitle}{{\em Proceedings of the
  1999 International Conference on Software Engineering, ICSE' 99, Los Angeles,
  CA, USA, May 16-22, 1999.}} \bibinfo{publisher}{{ACM}},
  \bibinfo{pages}{411--420}.
\newblock


\bibitem[\protect\citeauthoryear{Fagin, Halpern, Moses, and Vardi}{Fagin
  et~al\mbox{.}}{1995}]%
        {FHMV95}
\bibfield{author}{\bibinfo{person}{Ronald Fagin}, \bibinfo{person}{{Joseph Y.}
  Halpern}, \bibinfo{person}{Yoram Moses}, {and} \bibinfo{person}{Moshe~Y.
  Vardi}.} \bibinfo{year}{1995}\natexlab{}.
\newblock \bibinfo{booktitle}{{\em Reasoning about Knowledge.}}
\newblock \bibinfo{publisher}{MIT Press}.
\newblock


\bibitem[\protect\citeauthoryear{Finkbeiner and Schewe}{Finkbeiner and
  Schewe}{2005}]%
        {FS05}
\bibfield{author}{\bibinfo{person}{Bernd Finkbeiner} {and}
  \bibinfo{person}{Sven Schewe}.} \bibinfo{year}{2005}\natexlab{}.
\newblock \showarticletitle{Uniform Distributed Synthesis}. In
  \bibinfo{booktitle}{{\em 20th {IEEE} Symposium on Logic in Computer Science
  {(LICS} 2005), 26-29 June 2005, Chicago, IL, USA, Proceedings}}.
  \bibinfo{publisher}{{IEEE}}, \bibinfo{pages}{321--330}.
\newblock
\showDOI{%
\url{https://doi.org/10.1109/LICS.2005.53}}


\bibitem[\protect\citeauthoryear{Fisher, Henzinger, Nickovic, Piterman, Singh,
  and Vardi}{Fisher et~al\mbox{.}}{2011}]%
        {FHNPSV11}
\bibfield{author}{\bibinfo{person}{Jasmin Fisher}, \bibinfo{person}{{Thomas A.}
  Henzinger}, \bibinfo{person}{Dejan Nickovic}, \bibinfo{person}{Nir Piterman},
  \bibinfo{person}{{Anmol V.} Singh}, {and} \bibinfo{person}{{Moshe Y.}
  Vardi}.} \bibinfo{year}{2011}\natexlab{}.
\newblock \showarticletitle{{Dynamic Reactive Modules}}. In
  \bibinfo{booktitle}{{\em {CONCUR} 2011 - Concurrency Theory - 22nd
  International Conference, {CONCUR} 2011, Aachen, Germany, September 6-9,
  2011. Proceedings}}. \bibinfo{pages}{404--418}.
\newblock
\showDOI{%
\url{https://doi.org/10.1007/978-3-642-23217-6\_27}}


\bibitem[\protect\citeauthoryear{Gochet and Gribomont}{Gochet and
  Gribomont}{2006}]%
        {GochetG06}
\bibfield{author}{\bibinfo{person}{Paul Gochet} {and}
  \bibinfo{person}{E.~Pascal Gribomont}.} \bibinfo{year}{2006}\natexlab{}.
\newblock \showarticletitle{Epistemic logic}.
\newblock In \bibinfo{booktitle}{{\em Logic and the Modalities in the Twentieth
  Century}}. \bibinfo{pages}{99--195}.
\newblock
\showDOI{%
\url{https://doi.org/10.1016/S1874-5857(06)80028-2}}


\bibitem[\protect\citeauthoryear{Griesmayer and Lomuscio}{Griesmayer and
  Lomuscio}{2013}]%
        {GL13}
\bibfield{author}{\bibinfo{person}{Andreas Griesmayer} {and}
  \bibinfo{person}{Alessio Lomuscio}.} \bibinfo{year}{2013}\natexlab{}.
\newblock \showarticletitle{Model Checking Distributed Systems against
  Temporal-Epistemic Specifications}. In \bibinfo{booktitle}{{\em Formal
  Techniques for Distributed Systems - Joint {IFIP} {WG} 6.1 International
  Conference, {FMOODS/FORTE} 2013, Held as Part of the 8th International
  Federated Conference on Distributed Computing Techniques, DisCoTec 2013,
  Florence, Italy, 2013. Proceedings}}. \bibinfo{pages}{130--145}.
\newblock
\showDOI{%
\url{https://doi.org/10.1007/978-3-642-38592-6\_10}}


\bibitem[\protect\citeauthoryear{Gutierrez, Harrenstein, and
  Wooldridge}{Gutierrez et~al\mbox{.}}{2017}]%
        {GHW17}
\bibfield{author}{\bibinfo{person}{Julian Gutierrez}, \bibinfo{person}{Paul
  Harrenstein}, {and} \bibinfo{person}{Michael Wooldridge}.}
  \bibinfo{year}{2017}\natexlab{}.
\newblock \showarticletitle{{From Model Checking to Equilibrium Checking:
  Reactive Modules for Rational Verification}}.
\newblock \bibinfo{journal}{{\em Artif. Intell.\/}}  \bibinfo{volume}{248}
  (\bibinfo{year}{2017}), \bibinfo{pages}{123--157}.
\newblock
\showDOI{%
\url{https://doi.org/10.1016/j.artint.2017.04.003}}


\bibitem[\protect\citeauthoryear{Gutierrez, Perelli, and Wooldridge}{Gutierrez
  et~al\mbox{.}}{2018}]%
        {GPW18}
\bibfield{author}{\bibinfo{person}{Julian Gutierrez}, \bibinfo{person}{Giuseppe
  Perelli}, {and} \bibinfo{person}{Michael Wooldridge}.}
  \bibinfo{year}{2018}\natexlab{}.
\newblock \showarticletitle{Imperfect information in Reactive Modules games}.
\newblock \bibinfo{journal}{{\em Inf. Comput.\/}} \bibinfo{volume}{261},
  \bibinfo{number}{Part} (\bibinfo{year}{2018}), \bibinfo{pages}{650--675}.
\newblock


\bibitem[\protect\citeauthoryear{Halpern and Vardi}{Halpern and Vardi}{1989}]%
        {HV89}
\bibfield{author}{\bibinfo{person}{{Joseph Y.} Halpern} {and}
  \bibinfo{person}{{Moshe Y.} Vardi}.} \bibinfo{year}{1989}\natexlab{}.
\newblock \showarticletitle{The Complexity of Reasoning about Knowledge and
  Time. I. Lower Bounds}.
\newblock \bibinfo{journal}{{\em J. Comput. Syst. Sci.\/}}
  \bibinfo{volume}{38}, \bibinfo{number}{1} (\bibinfo{year}{1989}),
  \bibinfo{pages}{195--237}.
\newblock
\showDOI{%
\url{https://doi.org/10.1016/0022-0000(89)90039-1}}


\bibitem[\protect\citeauthoryear{Hamann}{Hamann}{2018}]%
        {swarmdesign}
\bibfield{author}{\bibinfo{person}{Heiko Hamann}.}
  \bibinfo{year}{2018}\natexlab{}.
\newblock \bibinfo{booktitle}{{\em Swarm Robotics - {A} Formal Approach}}.
\newblock \bibinfo{publisher}{Springer}.
\newblock
\showISBNx{978-3-319-74526-8}
\showDOI{%
\url{https://doi.org/10.1007/978-3-319-74528-2}}


\bibitem[\protect\citeauthoryear{Helsinger, Thome, and Wright}{Helsinger
  et~al\mbox{.}}{2004}]%
        {cougaar}
\bibfield{author}{\bibinfo{person}{Aaron Helsinger}, \bibinfo{person}{Michael
  Thome}, {and} \bibinfo{person}{Todd Wright}.}
  \bibinfo{year}{2004}\natexlab{}.
\newblock \showarticletitle{Cougaar: a scalable, distributed multi-agent
  architecture}. In \bibinfo{booktitle}{{\em 2004 IEEE International Conference
  on Systems, Man and Cybernetics (IEEE Cat. No.04CH37583)}},
  Vol.~\bibinfo{volume}{2}. \bibinfo{pages}{1910--1917 vol.2}.
\newblock
\showDOI{%
\url{https://doi.org/10.1109/ICSMC.2004.1399959}}


\bibitem[\protect\citeauthoryear{Karagiannis, Matthaiakis, Andronas, Filis,
  Giannoulis, Michalos, and Makris}{Karagiannis et~al\mbox{.}}{2019}]%
        {rms1}
\bibfield{author}{\bibinfo{person}{Panagiotis Karagiannis},
  \bibinfo{person}{Stereos~Alexandros Matthaiakis}, \bibinfo{person}{Dionisis
  Andronas}, \bibinfo{person}{Konstantinos Filis}, \bibinfo{person}{Christos
  Giannoulis}, \bibinfo{person}{George Michalos}, {and}
  \bibinfo{person}{Sotiris Makris}.} \bibinfo{year}{2019}\natexlab{}.
\newblock \showarticletitle{{Reconfigurable Assembly Station: A Consumer Goods
  Industry Paradigm}}.
\newblock \bibinfo{journal}{{\em {Procedia CIRP}\/}}  \bibinfo{volume}{81}
  (\bibinfo{year}{2019}), \bibinfo{pages}{1406--1411}.
\newblock
\showISSN{2212-8271}
\showDOI{%
\url{https://doi.org/10.1016/j.procir.2019.04.070}}
\newblock
\shownote{52nd CIRP Conference on Manufacturing Systems (CMS), Ljubljana,
  Slovenia, June 12-14, 2019.}


\bibitem[\protect\citeauthoryear{Kouvaros, Lomuscio, Pirovano, and
  Punchihewa}{Kouvaros et~al\mbox{.}}{2019}]%
        {KLPP19}
\bibfield{author}{\bibinfo{person}{Panagiotis Kouvaros},
  \bibinfo{person}{Alessio Lomuscio}, \bibinfo{person}{Edoardo Pirovano}, {and}
  \bibinfo{person}{Hashan Punchihewa}.} \bibinfo{year}{2019}\natexlab{}.
\newblock \showarticletitle{Formal Verification of Open Multi-Agent Systems}.
  In \bibinfo{booktitle}{{\em Proceedings of the 18th International Conference
  on Autonomous Agents and MultiAgent Systems, {AAMAS} '19, Montreal, QC,
  Canada, May 13-17, 2019}}. \bibinfo{pages}{179--187}.
\newblock


\bibitem[\protect\citeauthoryear{Kozen}{Kozen}{1983}]%
        {Koz83}
\bibfield{author}{\bibinfo{person}{Dexter Kozen}.}
  \bibinfo{year}{1983}\natexlab{}.
\newblock \showarticletitle{{Results on the Propositional mu-Calculus}}.
\newblock \bibinfo{journal}{{\em Theor. Comput. Sci.\/}}  \bibinfo{volume}{27}
  (\bibinfo{year}{1983}), \bibinfo{pages}{333--354}.
\newblock
\showDOI{%
\url{https://doi.org/10.1016/0304-3975(82)90125-6}}


\bibitem[\protect\citeauthoryear{Lomuscio, Qu, and Raimondi}{Lomuscio
  et~al\mbox{.}}{2017}]%
        {LQR17}
\bibfield{author}{\bibinfo{person}{Alessio Lomuscio}, \bibinfo{person}{Hongyang
  Qu}, {and} \bibinfo{person}{Franco Raimondi}.}
  \bibinfo{year}{2017}\natexlab{}.
\newblock \showarticletitle{{MCMAS:} an open-source model checker for the
  verification of multi-agent systems}.
\newblock \bibinfo{journal}{{\em {STTT}\/}} \bibinfo{volume}{19},
  \bibinfo{number}{1} (\bibinfo{year}{2017}), \bibinfo{pages}{9--30}.
\newblock


\bibitem[\protect\citeauthoryear{Maganha, Silva, and Ferreira}{Maganha
  et~al\mbox{.}}{2019}]%
        {rms2}
\bibfield{author}{\bibinfo{person}{Isabela Maganha}, \bibinfo{person}{Cristovao
  Silva}, {and} \bibinfo{person}{{Luis Miguel D. F.} Ferreira}.}
  \bibinfo{year}{2019}\natexlab{}.
\newblock \showarticletitle{{The Layout Design in Reconfigurable Manufacturing
  Systems: a Literature Review}}.
\newblock \bibinfo{journal}{{\em {The International Journal of Advanced
  Manufacturing Technology}\/}} (\bibinfo{year}{2019}).
\newblock
\showISSN{1433-3015}
\showDOI{%
\url{https://doi.org/10.1007/s00170-019-04190-3}}


\bibitem[\protect\citeauthoryear{Maoz and Ringert}{Maoz and Ringert}{2015}]%
        {MR15}
\bibfield{author}{\bibinfo{person}{Shahar Maoz} {and} \bibinfo{person}{{Jan
  Oliver} Ringert}.} \bibinfo{year}{2015}\natexlab{}.
\newblock \showarticletitle{{GR(1)} synthesis for {LTL} specification
  patterns}. In \bibinfo{booktitle}{{\em Proceedings of the 2015 10th Joint
  Meeting on Foundations of Software Engineering, {ESEC/FSE} 2015, Bergamo,
  Italy, August 30 - September 4, 2015}}. \bibinfo{publisher}{{ACM}},
  \bibinfo{pages}{96--106}.
\newblock


\bibitem[\protect\citeauthoryear{Mathews, Christensen, Ferrante, O'Grady, and
  Dorigo}{Mathews et~al\mbox{.}}{2010}]%
        {MathewsVCOBD15}
\bibfield{author}{\bibinfo{person}{Nithin Mathews},
  \bibinfo{person}{Anders~Lyhne Christensen}, \bibinfo{person}{Eliseo
  Ferrante}, \bibinfo{person}{Rehan O'Grady}, {and} \bibinfo{person}{Marco
  Dorigo}.} \bibinfo{year}{2010}\natexlab{}.
\newblock \showarticletitle{Establishing Spatially Targeted Communication in a
  Heterogeneous Robot Swarm}. In \bibinfo{booktitle}{{\em Proceedings of the
  9th International Conference on Autonomous Agents and Multiagent Systems:
  Volume 1}} {\em (\bibinfo{series}{AAMAS '10})}.
  \bibinfo{publisher}{International Foundation for Autonomous Agents and
  Multiagent Systems}, \bibinfo{pages}{939--946}.
\newblock
\showISBNx{978-0-9826571-1-9}


\bibitem[\protect\citeauthoryear{{{Michael J.} Fischer and {Richard E.}
  Ladner}}{{{Michael J.} Fischer and {Richard E.} Ladner}}{1979}]%
        {FL79}
\bibfield{author}{\bibinfo{person}{{{Michael J.} Fischer and {Richard E.}
  Ladner}}.} \bibinfo{year}{1979}\natexlab{}.
\newblock \showarticletitle{Propositional Dynamic Logic of Regular Programs}.
\newblock \bibinfo{journal}{{\em J. Comput. Syst. Sci.\/}}
  \bibinfo{volume}{18}, \bibinfo{number}{2} (\bibinfo{year}{1979}),
  \bibinfo{pages}{194--211}.
\newblock
\showDOI{%
\url{https://doi.org/10.1016/0022-0000(79)90046-1}}


\bibitem[\protect\citeauthoryear{Milner, Parrow, and Walker}{Milner
  et~al\mbox{.}}{1992}]%
        {pi1}
\bibfield{author}{\bibinfo{person}{Robin Milner}, \bibinfo{person}{Joachim
  Parrow}, {and} \bibinfo{person}{David Walker}.}
  \bibinfo{year}{1992}\natexlab{}.
\newblock \showarticletitle{A Calculus of Mobile Processes, {I}}.
\newblock \bibinfo{journal}{{\em Inf. Comput.\/}} \bibinfo{volume}{100},
  \bibinfo{number}{1} (\bibinfo{year}{1992}), \bibinfo{pages}{1--40}.
\newblock
\showDOI{%
\url{https://doi.org/10.1016/0890-5401(92)90008-4}}


\bibitem[\protect\citeauthoryear{Miyano and Hayashi}{Miyano and
  Hayashi}{1984}]%
        {mh84}
\bibfield{author}{\bibinfo{person}{Satoru Miyano} {and}
  \bibinfo{person}{Takeshi Hayashi}.} \bibinfo{year}{1984}\natexlab{}.
\newblock \showarticletitle{Alternating Finite Automata on omega-Words}.
\newblock \bibinfo{journal}{{\em Theor. Comput. Sci.\/}}  \bibinfo{volume}{32}
  (\bibinfo{year}{1984}), \bibinfo{pages}{321--330}.
\newblock
\showDOI{%
\url{https://doi.org/10.1016/0304-3975(84)90049-5}}


\bibitem[\protect\citeauthoryear{Mogavero, Murano, Perelli, and Vardi}{Mogavero
  et~al\mbox{.}}{2014}]%
        {MMPV14}
\bibfield{author}{\bibinfo{person}{Fabio Mogavero}, \bibinfo{person}{Aniello
  Murano}, \bibinfo{person}{Giuseppe Perelli}, {and} \bibinfo{person}{{Moshe
  Y.} Vardi}.} \bibinfo{year}{2014}\natexlab{}.
\newblock \showarticletitle{{Reasoning About Strategies: On the Model-Checking
  Problem.}}
\newblock \bibinfo{journal}{{\em TOCL\/}} \bibinfo{volume}{15},
  \bibinfo{number}{4}.
\newblock
\newblock
\shownote{doi:10.1145/2631917.}


\bibitem[\protect\citeauthoryear{Rubenstein, Ahler, Hoff, Cabrera, and
  Nagpal}{Rubenstein et~al\mbox{.}}{2014}]%
        {kilobot}
\bibfield{author}{\bibinfo{person}{Michael Rubenstein},
  \bibinfo{person}{Christian Ahler}, \bibinfo{person}{Nick Hoff},
  \bibinfo{person}{Adrian Cabrera}, {and} \bibinfo{person}{Radhika Nagpal}.}
  \bibinfo{year}{2014}\natexlab{}.
\newblock \showarticletitle{Kilobot: {A} low cost robot with scalable
  operations designed for collective behaviors}.
\newblock \bibinfo{journal}{{\em Robotics and Autonomous Systems\/}}
  \bibinfo{volume}{62}, \bibinfo{number}{7} (\bibinfo{year}{2014}),
  \bibinfo{pages}{966--975}.
\newblock
\showDOI{%
\url{https://doi.org/10.1016/j.robot.2013.08.006}}


\bibitem[\protect\citeauthoryear{Wang, Jiang, Chen, Yi, and Cui}{Wang
  et~al\mbox{.}}{2017}]%
        {APUS}
\bibfield{author}{\bibinfo{person}{Cheng Wang}, \bibinfo{person}{Jianyu Jiang},
  \bibinfo{person}{Xusheng Chen}, \bibinfo{person}{Ning Yi}, {and}
  \bibinfo{person}{Heming Cui}.} \bibinfo{year}{2017}\natexlab{}.
\newblock \showarticletitle{{APUS:} fast and scalable paxos on {RDMA}}. In
  \bibinfo{booktitle}{{\em Proceedings of the 2017 Symposium on Cloud
  Computing, SoCC 2017, Santa Clara, CA, USA, September 24-27, 2017}}.
  \bibinfo{pages}{94--107}.
\newblock
\showDOI{%
\url{https://doi.org/10.1145/3127479.3128609}}


\bibitem[\protect\citeauthoryear{Wong and Kress{-}Gazit}{Wong and
  Kress{-}Gazit}{2016}]%
        {WK16}
\bibfield{author}{\bibinfo{person}{Kai~Weng Wong} {and} \bibinfo{person}{Hadas
  Kress{-}Gazit}.} \bibinfo{year}{2016}\natexlab{}.
\newblock \showarticletitle{Need-based coordination for decentralized
  high-level robot control}. In \bibinfo{booktitle}{{\em 2016 {IEEE/RSJ}
  International Conference on Intelligent Robots and Systems, {IROS} 2016,
  Daejeon, South Korea, October 9-14, 2016}}. \bibinfo{pages}{2209--2216}.
\newblock
\showDOI{%
\url{https://doi.org/10.1109/IROS.2016.7759346}}


\bibitem[\protect\citeauthoryear{Wooldridge}{Wooldridge}{2002}]%
        {Woo02}
\bibfield{author}{\bibinfo{person}{Michael Wooldridge}.}
  \bibinfo{year}{2002}\natexlab{}.
\newblock \bibinfo{booktitle}{{\em {Introduction to Multiagent Systems.}}}
\newblock \bibinfo{publisher}{Wiley}.
\newblock


\end{thebibliography}

\end{document}